\documentclass[english,letter paper,12pt,reqno]{article}
\usepackage{etex} 
\usepackage{array}
\usepackage{varwidth}
\usepackage{bussproofs}
\usepackage{epigraph}
\usepackage{stmaryrd}
\usepackage{mathdots}
\usepackage{longtable}
\usepackage{tabu}
\usepackage{tikz-cd} 
\usepackage{amsmath, amscd, amssymb, mathrsfs, accents, amsfonts,amsthm}
\usepackage[all]{xy}
\usepackage{mathtools} 
\usepackage{tikz}
\usetikzlibrary{calc}

\AtEndDocument{\bigskip{\footnotesize%
  James Clift \par
  \textsc{Department of Mathematics, University of Melbourne} \par
  \textit{E-mail address}: \texttt{j.clift3@student.unimelb.edu.au} \par
  \vspace{0.3cm}
  Daniel Murfet \par
  \textsc{Department of Mathematics, University of Melbourne} \par  
  \textit{E-mail address}: \texttt{d.murfet@unimelb.edu.au} \par
}}

\newenvironment{mathprooftree}
  {\varwidth{.9\textwidth}\centering\leavevmode}
  {\DisplayProof\endvarwidth}
  
\DeclarePairedDelimiter\ket{\lvert}{\rangle}
\DeclarePairedDelimiterX\braket[2]{\langle}{\rangle}{#1 \delimsize\vert #2}
\DeclarePairedDelimiterX\inner[2]{\langle}{\rangle}{#1,#2}

\def\Ket#1{\left|#1\right\rangle}

\definecolor{Myblue}{rgb}{0,0,0.6}
\usepackage[a4paper,colorlinks,citecolor=Myblue,linkcolor=Myblue,urlcolor=Myblue,pdfpagemode=None]{hyperref}

\SelectTips{cm}{}

\newtheorem{theorem}{Theorem}[section]
\newtheorem{proposition}[theorem]{Proposition}
\newtheorem{lemma}[theorem]{Lemma}
\newtheorem{corollary}[theorem]{Corollary}

\newtheoremstyle{example}{\topsep}{\topsep}
	{}
	{}
	{\bfseries}
	{.}
	{2pt}
	{\thmname{#1}\thmnumber{ #2}\thmnote{ #3}}
	
	\theoremstyle{example}
	\newtheorem{definition}[theorem]{Definition}
	\newtheorem{example}[theorem]{Example}
	\newtheorem{remark}[theorem]{Remark}

\def\eval{\operatorname{ev}}

\def\res{\operatorname{Res}}

\def\Hom{\operatorname{Hom}}

\def\vacu{\ket{\emptyset}}
\def\be{\begin{equation}}
\def\ee{\end{equation}}

\DeclareMathOperator{\End}{End}

\DeclareMathOperator{\Spec}{Spec}

\DeclareMathOperator{\Sym}{Sym}

\def\inta{\bold{int}}
\def\comp{\underline{\textup{comp}}}
\def\contract{\;\lrcorner\;}
\newcommand{\HOM}{\operatorname{HOM}}
\newcommand{\op}{\text{op}}
\newcommand{\kk}{k} 
\newcommand{\ob}{\operatorname{ob}}
\newcommand{\id}{\text{id}}
\newcommand{\Sum}{\sum\limits}

\newcommand{\cCoalg}{\textbf{Coalg}}

\def\sC{\mathscr{C}}

\def\sV{\cat{V}}

\newcommand{\lxto}[2][3.5em]{\xrightarrow{\mathmakebox[#1]{#2}}}
\newcommand{\lxmapsto}[2][3.5em]{\xmapsto{\mathmakebox[#1]{#2}}}
\newcommand{\dntn}[1]{\llbracket #1 \rrbracket} 
\newcommand{\proofvdots}[1]{\overset{\displaystyle #1}{\rvdots}}

\makeatletter
\DeclareRobustCommand{\rvdots}{%
  \vbox{
    \baselineskip4\p@\lineskiplimit\z@
    \kern-\p@
    \hbox{}\hbox{.}\hbox{.}\hbox{.}
  }}
\makeatother

\begin{document}

\def\ScoreOverhang{1pt}

\def\Res{\res\!}
\newcommand{\ud}[1]{\operatorname{d}\!{#1}}
\newcommand{\Ress}[1]{\res_{#1}\!}
\newcommand{\cat}[1]{\mathcal{#1}}
\newcommand{\lto}{\longrightarrow}
\newcommand{\xlto}[1]{\stackrel{#1}\lto}
\newcommand{\mf}[1]{\mathfrak{#1}}
\newcommand{\md}[1]{\mathscr{#1}}
\newcommand{\church}[1]{\underline{#1}}
\newcommand{\prf}[1]{\underline{#1}}
\newcommand{\den}[1]{\llbracket #1 \rrbracket}
\def\l{\,|\,}
\def\sgn{\textup{sgn}}
\def\cont{\operatorname{cont}}
\def\counit{\varepsilon}
\def\inta{\textbf{int}}
\def\binta{\textbf{bint}}
\def\comp{\underline{\textup{comp}}}
\def\mult{\underline{\textup{mult}}}
\def\repeat{\underline{\textup{repeat}}}
\def\contract{\;\lrcorner\;}
\def\Coalg{\textbf{Coalg}}
\def\Alg{\textbf{Alg}}
\newcommand{\prom}{\operatorname{prom}}
\def\<{\langle} \def\>{\rangle}

\title{Cofree coalgebras and differential linear logic}
\author{James Clift, Daniel Murfet}

\maketitle

\begin{abstract} We prove that the semantics of intuitionistic linear logic in vector spaces which uses cofree coalgebras is also a model of differential linear logic, and that the Cartesian closed category of cofree coalgebras is a model of the simply-typed differential lambda calculus.
\end{abstract}

\tableofcontents

\setlength{\epigraphwidth}{0.82\textwidth}
\epigraph{In the discrete world of computing, there is no meaningful metric in which ``small" changes and ``small" effects go hand in hand, and there never will be.}{E.W.Dijkstra, \textit{On the cruelty of really teaching computer science}}

\section{Introduction}

The idea of taking derivatives of programs is an old one \cite[\S 2]{paige} with manifestations including automatic differentiation of algorithms computing real-valued functions \cite{autodiff} and incremental computation \cite{incdiff}. However, these approaches are limited to restricted classes of computations, and it is only recently with the development of the differential $\lambda$-calculus by Ehrhard-Regnier \cite{difflambda} and its refinement by differential linear logic \cite{blutecs, ehrhard-survey} that derivatives have been defined for general higher-order programs. These theories assign to each program $P$ another program $\partial P$, the derivative, which (in some sense) computes the infinitesimal change in the output of $P$ resulting from an infinitesimal change to its input. 

The connection between the Ehrhard-Regnier derivative and ordinary calculus is made explicit in the semantics of differential $\lambda$-calculus and differential linear logic, with the standard examples being the K\"othe and finiteness space semantics of Ehrhard \cite{ehrhard-kothe, ehrhard-finiteness} and the semantics of Blute-Ehrhard-Tasson \cite{blutecon} in convenient vector spaces. In this paper we explain how the simplest semantics of intuitionistic linear logic in vector spaces \cite{hyland, murfet_ll} is already a model of differential linear logic: tangent vectors and derivatives appear automatically when we use the cofree coalgebra to model the exponential.

Here is a sketch of the key point: let $\den{-}$ denote the natural semantics of linear logic in $\mathbb{C}$-vector spaces and suppose we are given a proof $\pi$ in linear logic computing a function from inputs of type $A$ to outputs of type $B$:
\[
\begin{mathprooftree}
\AxiomC{$\pi$}
\noLine\UnaryInfC{$\vdots$}
\def\extraVskip{5pt}
\noLine\UnaryInfC{${!} A \vdash B$\,.}
\end{mathprooftree}
\]
The space of inputs to $\den{\pi}$ is $\den{A}$, which for the sake of simplicity let us assume is finite-dimensional. A small change in the input starting from $P \in \den{A}$ is a tangent vector $v$ at the point $P$. We may identify $v$ with a vector in $\den{A}$, so that a representative infinitesimal curve segment of this tangent vector is
\[
(-1,1) \lto \den{A}\,, \quad t \longmapsto P + tv\,.
\]
A pair of vectors $(P, v)$ in $\den{A}$ is the same data as a linear map
\be\label{eq:linear_map_tangent_intro}
\xymatrix@C+2pc{
\mathbb{C}1^* \oplus \mathbb{C} \varepsilon^* \ar[r] & \den{A}
}
\ee
where the basis element $1^*$ is sent to $P$ and $\varepsilon^*$ is sent to $v$. The naming of these basis elements is chosen so that we may identify $\mathbb{C} 1^* \oplus \mathbb{C} \varepsilon^*$ with the dual vector space $(\mathbb{C}[\varepsilon]/\varepsilon^2)^*$. As the dual of an algebra this vector space is canonically equipped with the structure of a coalgebra, with comultiplication $1^* \mapsto 1^* \otimes 1^*$ and $\varepsilon^* \mapsto 1^* \otimes \varepsilon^* + \varepsilon^* \otimes 1^*$. Our tangent vector has now been uniquely associated to a linear map from the underlying vector space of the coalgebra $(\mathbb{C}[\varepsilon]/\varepsilon^2)^*$ to $\den{A}$. 

However, there is a \emph{universal} pair consisting of a coalgebra together with a linear map to $\den{A}$, namely Sweedler's cofree coalgebra ${!} \den{A}$ together with its canonical linear map ${!} \den{A} \lto \den{A}$. The universal property means that the linear map \eqref{eq:linear_map_tangent_intro} encoding our tangent vector lifts uniquely to a morphism of coalgebras
\be\label{eq:toucan}
\xymatrix@C+2pc{
(\mathbb{C}[\varepsilon]/\varepsilon^2)^* \ar[r] & {!} \den{A}\,.
}
\ee
The function assigning to $(P, v)$ this morphism of coalgebras is a bijection between pairs consisting of a point in $\den{A}$ and a tangent vector, and morphisms of coalgebras \eqref{eq:toucan}. 

Here is where $\pi$ enters the story. In the semantics of linear logic in vector spaces the denotation of ${!} A$ is ${!} \den{A}$, and the denotation of $\pi$ is a linear map $\den{\pi}: \den{{!} A} \lto \den{B}$ which lifts by the universal property to a morphism of coalgebras ${!} \den{A} \lto {!} \den{B}$. This lifting may be composed with \eqref{eq:toucan} to give a morphism of coalgebras
\be\label{eq:toucan2}
\xymatrix@C+2pc{
(\mathbb{C}[\varepsilon]/\varepsilon^2)^* \ar[r] & {!} \den{A} \ar[r] & {!} \den{B}\,.
}
\ee
Working the bijection between tangent vectors and coalgebra morphisms in reverse, we see that \eqref{eq:toucan2} uniquely specifies a point in $\den{B}$ and a tangent vector. The point in $\den{B}$ is just the output of the algorithm $\pi$ on the given input, while the tangent vector gives the infinitesimal variation of the output, when the input is varied in the direction of $v$.

The formal statement is that for any algebraically closed field $k$ of characteristic zero the semantics of intuitionistic linear logic in $k$-vector spaces defined using cofree coalgebras is a model of differential linear logic (Theorem \ref{main_theorem}). We refer to this as the \emph{Sweedler semantics}, since the explicit description of this universal coalgebra is due to him  \cite{sweedler,murfet_ll}. The proof is elementary and we make no claim here to technical novelty; the link between the symmetric coalgebra and differential calculus is well-known. Perhaps our main contribution is to give several detailed examples showing how to compute these derivatives (Section \ref{section:examples}) and to check that this model of differential linear logic gives rise to a model of differential lambda calculus in the category of cofree coalgebras (Theorem \ref{theorem:cofree_difflambda}).
\\

We conclude this introduction with a sketch of one such example and a comparison of our work to other semantics of differential linear logic. To elaborate a little more on the notation: for any type $A$ of linear logic (which for us has only connectives $\otimes, \multimap, !$) there is a vector space $\den{A}$, and for any proof $\pi$ of $A \vdash B$ there is a linear map $\den{\pi}: \den{A} \lto \den{B}$. In particular every proof $\xi$ of type $A$ has a denotation $\den{\xi} \in \den{A}$, and the promotion of $\xi$ has for its denotation a vector $\ket{\emptyset}_{\den{\xi}} \in \den{!A}$, see \cite[\S 5.3]{murfet_ll}.

For any binary sequence $S \in \{0,1\}^*$ there is an encoding of $S$ as a proof $\underline{S}$ of type
\[
\textbf{bint}_A = {!}(A \multimap A) \multimap \big({!}(A \multimap A) \multimap (A \multimap A)\big)\,.
\]
Repetition of sequences can be encoded as a proof
\[
\begin{mathprooftree}
\AxiomC{$\prf{\mathrm{repeat}}$}
\noLine\UnaryInfC{$\vdots$}
\def\extraVskip{5pt}
\noLine\UnaryInfC{${!} \textbf{bint}_A \vdash \textbf{bint}_A$\,.}
\RightLabel{\scriptsize $\multimap R$}
\end{mathprooftree}
\]
The denotation is a linear map $\den{ {!}\textbf{bint}_A } \lto \den{ \textbf{bint}_A }$ sending $\ket{\emptyset}_{\den{\underline{S}}}$ to $\den{\underline{SS}}$. The derivative of $\prf{\mathrm{repeat}}$ according to the theory of differential linear logic is another a proof
\[
\begin{mathprooftree}
\AxiomC{$\partial\, \prf{\mathrm{repeat}}$}
\noLine\UnaryInfC{$\vdots$}
\def\extraVskip{5pt}
\noLine\UnaryInfC{${!} \textbf{bint}_A, \textbf{bint}_A \vdash \textbf{bint}_A$\,}
\RightLabel{\scriptsize $\multimap R$}
\end{mathprooftree}
\]
which can be derived from $\prf{\mathrm{repeat}}$ by new deduction rules called codereliction, cocontraction and coweakening (see Section \ref{section:coder}). We prove in Section \ref{section:bint} that the denotation of this derivative in the Sweedler semantics is the linear map
\begin{gather*}
\den{\partial\, \prf{\mathrm{repeat}}}: \den{ {!}\textbf{bint}_A } \otimes \den{ \textbf{bint}_A} \lto \den{ \textbf{bint}_A }\,,\\
\ket{\emptyset}_{\den{\underline{S}}} \otimes \den{\underline{T}} \longmapsto \den{\underline{ST}} + \den{\underline{TS}}
\end{gather*}
whose value on the tensor $\ket{\emptyset}_{\den{\underline{S}}} \otimes \den{\underline{T}}$ we interpret as the derivative of the repeat program at the sequence $S$ in the direction of the sequence $T$. This can be justified informally by the following calculation using an infinitesimal $\varepsilon$
\begin{align*}
(S + \varepsilon T)( S + \varepsilon T) = SS + \varepsilon( ST + TS ) + \varepsilon^2 TT,
\end{align*}
which says that varying the sequence infinitesimally from $S$ in the direction of $T$ causes a variation of the repetition in the direction of $ST + TS$. 
\\

The Sweedler semantics is far from the first semantics of differential linear logic: basic examples include the categories of sets and relations \cite[\S 2.5.1]{blutecs} and suplattices \cite[\S 2.5.2]{blutecs}. The examples of K\"othe and finiteness spaces \cite{ehrhard-kothe, ehrhard-finiteness} and convenient vector spaces \cite{blutecon,frolicher} have already been mentioned. These papers explain that the geometric ``avatar'' of the exponential connective of linear logic is the functor sending a space $X$ to the space of distributions on $X$ (for a precise statement, see Remark \ref{remark:distr}). This remarkable analogy between logic and geometry deserves further study. 

Conceptually the Sweedler semantics is similar to these examples in that the exponential is modelled by a space of distributions (with finite support) but it is purely algebraic and there are simple explicit formulas for all the structure maps, which makes it suitable for concrete calculations with proof denotations. Moreover in the algebraic approach the differential structure emerges naturally from the exponential structure, rather than being ``baked in''. The downside is that the smoothness of proof denotations in our semantics is obscured; in particular, in the case $k = \mathbb{C}$ some extra work is required to see the relation between our differential structure and the derivatives in the usual sense.
\\

\emph{Acknowledgements.} Thanks to Kazushige Terui, who stimulated this project by asking if the cofree coalgebra gave a model of differential linear logic.

\section{Sweedler semantics}

We review the Sweedler semantics of multiplicative exponential intuitionistic linear logic (henceforth simply \emph{linear logic}) in the category of vector spaces $\cat{V}$ over an algebraically closed field\footnote{The formal theory of coalgebras is simpler over algebraically closed fields, which explains why we use $k = \mathbb{C}$ in our examples, but this is not really important: one could work over $k = \mathbb{R}$ by taking $\mathbb{C}$-points into account in the explicit description of the cofree coalgebra.}
 $k$ of characteristic zero (e.g. $k = \mathbb{C}$) in Section \ref{section:sweedler_sem}. This was introduced in \cite{hyland} and revisited in \cite{murfet_ll} with a focus on explicit formulas for the involved structures \cite{murfet_coalg}. For background material on linear logic and its semantics see \cite{girard_llogic, girard_prooftypes, mellies}. The multiplicative connectives $\otimes$ and $ \multimap$ have the obvious interpretation; the only nontrivial ingredient in the Sweedler semantics is the cofree coalgebra which interprets the exponential. We begin this section with a review of cofree coalgebras (Section \ref{section:cofree_coalg_review}) and how to think about points (Section \ref{section:group_like_elements}) and tangent vectors (Section \ref{section:tangent_vectors}) in coalgebraic language.

Let $\Alg_k$ denote the category of commutative unital $k$-algebras and $\Coalg_k$ the category of cocommutative counital coalgebras. Unless otherwise indicated, all algebras are commutative and unital, and all coalgebras are cocommutative and counital. Throughout $\otimes = \otimes_k$ and $\Delta, \varepsilon$ denote respectively the comultiplication and counit of a coalgebra.

\subsection{Cofree coalgebras}\label{section:cofree_coalg_review}

The following construction is from \cite[Chapter VI]{sweedler}:

\begin{definition} The \emph{Hopf dual} or continuous linear dual $A^{\circ}$ of an algebra $A$ is the subspace of $A^* = \Hom_k(A,k)$ consisting of linear maps $A \lto k$ which factor as a composite
\[
\xymatrix{
A \ar[r] & A/I \ar[r] & k
}
\]
where $I \subseteq A$ is an ideal, the first map is the quotient and $A/I$ is finite-dimensional.
\end{definition}

The dual $A^{\circ}$ is sometimes denoted $\Hom^{\cont}(A,k)$, as for example in \cite{murfet_coalg}. In \cite[Lemma 6.0.1]{sweedler} it is proven that the canonical injective map $A^* \otimes A^* \lto (A \otimes A)^*$ identifies the subspace $A^{\circ} \otimes A^{\circ}$ with $(A \otimes A)^{\circ}$ and that the dual of the multiplication
\[
M^*: A^* \lto (A \otimes A)^*
\]
satisfies $M^*(A^{\circ}) \subseteq (A \otimes A)^{\circ}$. Identifying the codomain with $A^{\circ} \otimes A^{\circ}$ defines a linear map
\[
\Delta: A^{\circ} \lto A^{\circ} \otimes A^{\circ}
\]
and in this way $(A^{\circ}, \Delta, \varepsilon)$ is a cocommutative coalgebra \cite[Proposition 6.0.2]{sweedler} where the counit $\varepsilon: A^{\circ} \lto k$ is evaluation at the identity $1 \in A$. Clearly if $A$ is finite-dimensional then $A^{\circ} = A^*$. The fundamental theorem about the Hopf dual is:

\begin{theorem}[Sweedler]\label{theorem:fund_adjunc} Given an algebra $A$ and coalgebra $C$, there is a natural bijection
\[
\emph{\Alg}_k( A, C^* ) \cong \emph{\Coalg}_k( C, A^{\circ} )\,.
\]
\end{theorem}
\begin{proof}
See \cite[Theorem 6.0.5]{sweedler}.
\end{proof}

An important example is the Hopf dual of the symmetric algebra $A = \Sym(V^*)$ over a finite-dimensional vector space $V$. Suppose $V$ has basis $e_1,\ldots,e_n$ with dual basis $x_i = e_i^*$. In this case the linear map
\begin{gather*}
\eta: k \lto V \otimes V^* \lto V \otimes \Sym(V^*)\\
1 \mapsto \sum_{i=1}^n e_i \otimes x_i
\end{gather*}
which is independent of the choice of basis, gives rise to a linear map
\[
c_V: \Sym(V^*)^{\circ} \lto V
\]
which sends $\theta \in \Sym(V^*)^{\circ}$ to the vector $c_V(\theta)$ computed by the composite
\[
\xymatrix@C+1pc{
k \ar[r]^-{\eta} & V \otimes \Sym(V^*) \ar[r]^-{1 \otimes \theta} & V \otimes k \cong V
}\,.
\]
The pair $(\Sym(V^*)^{\circ}, c_V)$ is the cofree coalgebra generated by $V$, more precisely:

\begin{theorem} For a finite-dimensional vector space $V$ the pair $(\Sym(V^*)^{\circ},c_V)$ is universal among pairs consisting of a (cocommutative) coalgebra and a linear map from that coalgebra to $V$, in the sense that the map
\begin{gather*}
\emph{\Coalg}_k(C, \Sym(V^*)^{\circ}) \lto \Hom_{k}(C, V)\,,\\
\varphi \longmapsto c_V \circ \varphi
\end{gather*}
is a bijection for any coalgebra $C$.
\end{theorem}
\begin{proof}
See \cite[Theorem 6.4.1, 6.4.3]{sweedler} or \cite[Theorem 2.20]{murfet_coalg}. Here is a sketch of the proof: since any coalgebra is a colimit of finite-dimensional sub-coalgebras \cite[Theorem 2.2.1]{sweedler} we can reduce to the case of $C$ finite-dimensional, where by Theorem \ref{theorem:fund_adjunc}
\begin{align*}
\Coalg_k( C, \Sym(V^*)^{\circ} ) &\cong \Alg_k( \Sym(V^*), C^* )\\
&\cong \Hom_k( V^*, C^* )\\
&\cong \Hom_k( C, V)
\end{align*}
as claimed.
\end{proof}

So much is immediate from \cite{sweedler}. However from the point of view of having a semantics of linear logic (or differential linear logic) in which one can actually do calculations, it is essential to have an \emph{explicit} description of $\Sym(V^*)^{\circ}$ and $c_V$. Providing such a description was the purpose of \cite{murfet_coalg} and we now give a (partially new) exposition of the relevant facts.
\\

In the following let $V$ be a finite-dimensional vector space. 

\begin{definition}
For $P \in V$ we define the linear map
\be
\Psi_P: \Sym(V) \lto \Sym(V^*)^{\circ}
\ee
using a choice of basis $e_1,\ldots,e_n$ of $V$ by
\[
\Psi_P( e_1^{a_1} \cdots e_n^{a_n} )(f) = \frac{\partial^{a_1}}{\partial x_1^{a_1}} \cdots \frac{\partial^{a_n}}{\partial x_n^{a_n}}(f)\Big\vert_{x_1=P_1,\ldots,x_n=P_n}
\]
where $x_i = e_i^*$ is the dual basis. Writing $\Sym_P(V)$ for a copy of $\Sym(V)$
\[
\Psi: \bigoplus_{P \in V} \Sym_P(V) \lto \Sym(V^*)^{\circ}
\]
is the linear map with $\Psi_P$ as its components.
\end{definition}

There are two things that need to be checked, for this $\Psi$ to be well-defined:

\begin{lemma} $\Psi$ is independent of the choice of basis used to define it.
\end{lemma}
\begin{proof}
A change of coordinates affects $e_i^{a_i}$ in the same way as it affects $\frac{\partial^{a_i}}{\partial x_i^{a_i}}$.
\end{proof}

\begin{lemma}\label{lemma:psilandsin} Given $P \in V$ and $a_1,\ldots,a_n \ge 0$, the functional
\[
\Psi_P(e_1^{a_1} \cdots e_n^{a_n}) \in \Sym(V^*)^*
\]
belongs to the subspace $\Sym(V^*)^{\circ}$.
\end{lemma}
\begin{proof}
We prove $\theta = \Psi_P(e_1^{a_1} \cdots e_n^{a_n})$ vanishes on $(x_1-P_1,\ldots,x_n-P_n)^{\sum_i a_i + 1}$. Suppose given a monomial $f = (x_1-P_1)^{b_1} \cdots (x_n-P_n)^{b_n}$ with $\sum_i b_i > \sum_i a_i$. Then the derivative of $f$ involved in $\theta(f)$ will be divisible by some $x_i - P_i$ and so vanishes at $P$.
\end{proof}


\begin{theorem} $\Psi$ is an isomorphism of vector spaces.
\end{theorem}
\begin{proof}
Set $R = \Sym(V^*)$ and given $P \in V$ let $\mf{m}_P = (x_1 - P_1,\ldots,x_n-P_n)$ denote the associated maximal ideal of $R$. Using the Chinese remainder theorem (see the proof of \cite[Lemma A.1]{murfet_coalg}) it is easy to see that for any functional $\theta \in R^{\circ}$ there is a unique $P \in V$ such that for some $j > 0$ the map $\theta$ factors as
\[
R \lto R/\mf{m}_P^j \lto k\,.
\]
Since $R_{\mf{m}_P} / \mf{m}_P^j R_{\mf{m}_P} \cong R/\mf{m}_P^j$ we have
\[
(R_{\mf{m}_P})^{\circ} = \varinjlim_{j > 0} \Hom_k( R/\mf{m}_P^j, k )\,.
\]
One way to restate the consequence of the Chinese remainder theorem is that every $\theta \in R^{\circ}$ belongs to $(R_{\mf{m}_P})^{\circ}$ for a unique $P$, that is, there is an isomorphism
\[
\xymatrix@C+2pc{
\bigoplus_{P \in V} (R_{\mf{m}_P})^{\circ} \ar[r]^-{\cong} & R^{\circ}\,.
}
\]
The proof of Lemma \ref{lemma:psilandsin} shows that the image of $\Psi_P$ lies in the subspace $(R_{\mf{m}_P})^{\circ} \subseteq R^{\circ}$, so it suffices to show that the map
\[
\Psi_P: \Sym_P(V) \lto (R_{\mf{m}})^{\circ}
\]
is a bijection. But the left hand side is a direct limit of subspaces $\Sym_P(V)_{\le j}$ spanned by monomials of degree $\le j$ and the right hand side is a direct limit of subspaces $\Hom_k( R/\mf{m}^{j+1}, k)$. Moreover there is a commutative diagram
\[
\xymatrix@C+2pc{
\Sym_P(V) \ar[r]^-{\Psi_P} & (R_{\mf{m}})^{\circ}\\
\Sym_P(V)_{\le j} \ar[u]\ar[r]_-{\Psi^j_P} & \Hom_k(R/\mf{m}^{j+1},k) \ar[u]
}
\]
where the vertical maps are inclusions. So it suffices to prove $\Psi_P^j$ is an isomorphism for each $j$. But this is clearly true for $j = 0$ and for $j > 0$ we proceed by induction. We have an exact sequence
\[
\xymatrix{
0 \ar[r] & \mf{m}^j / \mf{m}^{j+1} \ar[r] & R / \mf{m}^{j+1} \ar[r] & R/\mf{m}^j \ar[r] & 0
}
\]
and hence a commutative diagram with exact rows
\[
\xymatrix@C-1pc{
0 \ar[r] & \Hom_k(R/\mf{m}^j, k) \ar[r] & \Hom_k( R/\mf{m}^{j+1}, k) \ar[r] & \Hom_k(\mf{m}^j/\mf{m}^{j+1},k) \ar[r] & 0\\
0 \ar[r] & \Sym_P(V)_{\le j-1} \ar[u]^{\cong} \ar[r] & \Sym_P(V)_{\le j} \ar[u]^{\Psi^j_P} \ar[r] & \Sym_P(V)_{\le j}/\Sym_P(V)_{\le j-1} \ar[u]_-{\Bar{\Psi}_P^j}\ar[r] & 0
}
\]
where $\bar{\Psi}^j_P$ is the induced map on the quotients, and the leftmost vertical map is an isomorphism by the inductive hypothesis. So it suffices to prove that $\Bar{\Psi_P^j}$ is an isomorphism (by the Five Lemma). But the domain and codomain both pick out ``monomials'' of degree $j$, in one case by removing from monomials of degree $\le j$ all those of degree $\le j - 1$ and in the other case by removing from monomials of degree $\ge j$ all those of degree $\ge j + 1$.

More formally, we may directly calculate that
\begin{align*}
\Bar{\Psi}_P^j( e_1^{a_1} \cdots e_n^{a_n} ) &= \sum_{b_1 + \cdots + b_n = j} \frac{\partial^{a_1}}{\partial x_1^{a_1}} \cdots \frac{\partial^{a_n}}{\partial x_n^{a_n}}\big( \omega_{\bold{b}} \big)\Big\vert_P \cdot \omega_{\bold{b}}^*\\
&= a_1{!} \cdots a_n{!} \cdot \omega_{\bold{a}}^*
\end{align*}
where for $\bold{b} = (b_1,\ldots,b_n)$ we write $\omega_{\bold{b}} = \prod_{j=1}^n (x_j-P_j)^{b_j}$ which under the restriction $|\bold{b}| \le j$ give a $k$-basis for $\mf{m}^j$ with dual basis $\omega_{\bold{b}}^*$.
\end{proof}

\begin{remark} In the notation of \cite{murfet_coalg} the map $\Psi_P$ is the composite of the isomorphisms in \cite[Lemma 2.12]{murfet_coalg} and \cite[Theorem 2.6]{murfet_coalg}
\[
\xymatrix{
\Sym_P(V) \ar[r]^-{\cong} & \operatorname{LC}(V, P) \ar[r]^-{\cong} & \Sym(V^*)_{\mf{m}_P}^{\circ}\\
}
\]
defined by
\[
e_1^{a_1} \cdots e_n^{a_n} \longmapsto a_1{!} \cdots a_n{!} \left[ \frac{f}{z_1^{a_1}, \ldots, z_n^{a_n}} \frac{\ud{\underline{z}}}{\underline{z}} \right] \longmapsto \frac{\partial^{a_1}}{\partial x_1^{a_1}} \cdots \frac{\partial^{a_n}}{\partial x_n^{a_n}}(-)\Big\vert_P
\]
where $z_i = x_i - P_i$.
\end{remark}

For any vector space $V$ (not necessarily finite-dimensional) the underlying vector space of the symmetric algebra $\Sym(V)$ is naturally equipped with the structure of a coalgebra (the \emph{symmetric coalgebra}) with comultiplication $\Delta$ defined by
\[
\Delta(v_1 \cdots v_n) = \sum_{I \subseteq \{1,\ldots,n\}} v_I \otimes v_{I^c}
\]
where $v_i \in V$ for $1 \le i \le n$ and for $I \subseteq \{1,\ldots,n\}$ we denote by $v_I$ the tensor which is the product in $\Sym(V)$ of the set $\{ v_i \l i \in I \}$. By convention if $I = \emptyset$ then $v_I = 1$. The counit $\varepsilon: \Sym(V) \lto k$ satisfies $\varepsilon(1) = 1$ and vanishes on monomials of positive degree; for the details see Bourbaki \cite[III \S 11]{bourbaki} (our coalgebras are their coassociative counital \emph{cogebras}).

\begin{proposition} $\Psi$ is an isomorphism of coalgebras.
\end{proposition}
\begin{proof}
It suffices to show that
\[
\Psi^j_P: \Sym_P(V)_{\le j} \lto Hom_k( R/\mf{m}_P^{j+1}, k)
\]
is a morphism of coalgebras, where $\Sym_P(V)_{\le j}$ is a subcoalgebra of the symmetric coalgebra  $\Sym_P(V)$ and $\Hom_k( R/\mf{m}_P^{j+1}, k)$ is given the coalgebra structure as the dual of the finite-dimensional algebra $R/\mf{m}_P^{j+1}$. Given $\theta: R/\mf{m}_P^{j+1} \lto k$ we have
\[
\Delta(\theta) = \sum_{|\bold{b}| \le j, |\bold{b}'| \le j} \theta( \omega_{\bold{b}} \cdot \omega_{\bold{b}'} ) \cdot \omega_{\bold{b}}^* \otimes \omega_{\bold{b}'}^*
\]
where $\omega_{\bold{b}} = \prod_{i=1}^n (x_i-P_i)^{b_j}$ and $|\bold{b}| = \sum_i b_i$. Hence
\[
\Delta( \omega_{\bold{c}}^* ) = \sum_{\bold{b} + \bold{b}' = \bold{c}} \omega_{\bold{b}}^* \otimes \omega_{\bold{b}'}^*\,.
\]
Now we have already calculated that
\[
\Psi_P( e_1^{a_1} \cdots e_n^{a_n} ) = a_1{!} \cdots a_n{!} \cdot \omega_{\bold{a}}^*
\]
so we have
\begin{align*}
(\Psi_P \otimes \Psi_P)\Delta(e_1^{a_1} \cdots e_n^{a_n}) &= (\Psi_P \otimes \Psi_P) \sum_{\bold{b} + \bold{b}' = \bold{a}} \binom{a_1}{b_1} \cdots \binom{a_n}{b_n} e_1^{b_1} \cdots e_n^{b_n} \otimes e_1^{b_1'} \cdots e_n^{b_n'}\\
&= \sum_{\bold{b} + \bold{b}' = \bold{a}} \binom{a_1}{b_1} \cdots \binom{a_n}{b_n} b_1{!} \cdots b_n{!} (b_1'){!} \cdots (b_n'){!} \omega_{\bold{b}}^* \otimes \omega_{\bold{b}'}^*\\
&= a_1{!} \cdots a_n{!} \sum_{\bold{b} + \bold{b}' = \bold{a}} \omega_{\bold{b}}^* \otimes \omega_{\bold{b}'}^*\\
&= \Delta \Psi_P(e_1^{a_1} \cdots e_n^{a_n})
\end{align*}
proving the claim. The compatibility of the counits is clear.
\end{proof}

The vector space $\Sym_P(V) = \bigoplus_{i \ge 0} \Sym_P^i(V)$ is graded, where $\Sym_P^i(V)$ is the image in the symmetric algebra of $V^{\otimes i}$. The projection from this graded vector space to its components $k \oplus V$ of degree $\le 1$, followed by the map $(\lambda, v) \mapsto \lambda P + v$ defines
\[
\xymatrix{
\Sym_P(V) \ar@{->>}[r] & k \oplus V \ar[rr]^-{\left(\begin{smallmatrix} P, \,1_V \end{smallmatrix}\right)} & & V
}
\]
and as $P$ varies these maps give the components of the linear map
\[
d_V: \bigoplus_P \Sym_P(V) \lto V\,, \qquad d_V|_{\Sym_P(V)}(v_0,v_1,v_2,\ldots) = v_0 P + v_1\,.
\]
Recall the linear map $c_V$, which we may compute in our basis to be
\[
c_V: \Sym(V^*)^{\circ} \lto V\,, \qquad c_V(\theta) = \sum_i \theta(x_i) e_i\,.
\]

\begin{lemma}\label{lemma:triangle_com} For a finite-dimensional vector space $V$ the diagram
\[
\xymatrix@C+1pc{
\bigoplus_P \Sym_P(V) \ar[dr]_-{d_V}\ar[rr]^-{\Psi}_-{\cong} & & \Sym(V^*)^{\circ} \ar[dl]^-{c_V}\\
& V
}
\]
commutes, and hence the pair $\big( \bigoplus_P \Sym_P(V), d_V \big)$ is also universal among pairs consisting of a (cocommutative) coalgebra and a linear map from that coalgebra to $V$.
\end{lemma}
\begin{proof}
We have $c_V \Psi_P(1) = \sum_i P_i e_i = P$ and
\begin{align*}
c_V \Psi_P(e_j) = \sum_i \Psi_P(e_j)(x_i) e_i = \sum_i \frac{\partial}{\partial x_j}(x_i) e_i = e_j
\end{align*}
whereas for a monomial $m \in \Sym_P(V)$ of degree $> 1$ both legs of the diagram vanish.
\end{proof}

In summary: there is a linear map $d_V$ from the coproduct of copies of the symmetric coalgebra indexed by the points of $V$. By the universal property of the Hopf dual, there is a unique morphism of coalgebras $\Psi$ making the diagram of Lemma \ref{lemma:triangle_com} commute, and what we have done in the above is compute explicitly this unique morphism of coalgebras in terms of differential operators, and prove that it is an isomorphism. 

In the above we have focused on finite-dimensional vector spaces $V$ because in this case the role of differential operators is most transparent. But the symmetric coalgebra $\Sym(V)$ is defined for any vector space $V$, with the same formulas for $\Delta, \varepsilon$ as those given above, and the coproduct $\bigoplus_{P \in V} \Sym_P(V)$ is therefore still a coalgebra. The linear map $d_V$ defined above remains well-defined, even if $V$ is infinite-dimensional.

\begin{proposition}\label{prop:bangV} For any vector space $V$ (not necessarily finite-dimensional) the pair
\[
{!} V := \big( \bigoplus_{P \in V} \Sym_P(V), d_V \big)
\]
is universal among pairs consisting of a cocommutative coalgebra and a linear map to $V$.
\end{proposition}
\begin{proof}
This is already implicit in Sweedler, as explained in \cite[Appendix B]{murfet_coalg}, but we give here another argument. Let $C$ be a coalgebra and $\varphi: C \lto V$ a linear map. We write $C$ as a direct limit of its finite-dimensional sub-coalgebras $\{ C_i \}_{i \in I}$ \cite[Theorem 2.2.1]{sweedler}. For each $i \in I$ let $V_i$ be the finite-dimensional subspace $\varphi(C_i)$ of $V$ and $\varphi_i: C_i \lto V_i$ be the restriction of $\varphi$. By Lemma \ref{lemma:triangle_com} there is a unique morphism of coalgebras $\Phi_i$ making
\[
\xymatrix@C+2pc{
C_i \ar[r]^-{\Phi_i} \ar[dr]_-{\varphi_i} & \bigoplus_{P \in V_i} \Sym_P(V_i) \ar[d]^{d_{V_i}}\\
& V_i
}
\]
commute. The inclusions $V_i \subseteq V$ induce morphisms of algebras $\Sym_P(V_i) \lto \Sym_P(V)$ which are easily checked to be injective morphisms of coalgebras. We let $\Phi'_i$ denote the composite with the direct sum of these inclusions:
\[
\xymatrix{
C_i \ar[r]^-{\Phi_i} & \bigoplus_{P \in V_i} \Sym_P(V_i) \ar[r] & \bigoplus_{P \in V_i} \Sym_P(V) \subseteq \bigoplus_{P \in V} \Sym_P(V)
}\,.
\]
These morphisms are compatible with inclusions $C_i \subseteq C_j$, and so induce a morphism of coalgebras $\Phi: C \lto \bigoplus_{P \in V} \Sym_P(V)$ which satisfies $d_V \circ \Phi = \varphi$. To show that $\Phi$ is unique with this property, let $\Phi'$ be some other morphism of coalgebras satisfying $d_V \circ \Phi' = \varphi$. For each $i \in I$ the restriction $\Phi'|_{C_i}$ factors as a morphism of coalgebras through
\[
\Sym_{P_1}(V) \oplus \cdots \oplus \Sym_{P_n}(V)
\]
for some finite set of points $P_1,\ldots,P_n \in V$ (depending on $i$) since $C_i$ is finite-dimensional. Moreover $\Phi'|_{C_i}$ must factor further through a finite-dimensional subspace of this sum, so there exists a finite-dimensional subspace $W_i \subseteq V$ such that $P_1,\ldots,P_n \in W_i$ and $\Phi'|_{C_i}$ factors as follows:
\[
\xymatrix{
C_i \ar[rr]^{\Phi'|_{C_i}}\ar[dr]_-{\Psi_i} & & \bigoplus_{P \in V} \Sym_P(V)\\
& \bigoplus_{P \in W_i} \Sym_{P_j}(W_i) \ar[ur]
}
\]
We may without loss of generality assume $V_i \subseteq W_i$ so that this factorisation $\Psi_i$ and the composite of $\Phi_i$ with the inclusion into $\bigoplus_{P \in W_i} \Sym_P(W_i)$ are two morphisms of coalgebras which agree when post-composed with $d_{W_i}$ (since $d_V \circ \Phi' = \varphi$). By the universal property, we conclude that they agree, and hence $\Phi'|_{C_i} = \Phi|_{C_i}$. It follows that $\Phi' = \Phi$.
\end{proof}

\begin{definition}\label{defn:cofree_coalg_gen}
The pair $({!} V, d_V)$ is referred to as the \emph{cofree coalgebra generated by $V$}.
\end{definition}

The cofree coalgebra is used to model the exponential connective of linear logic in the Sweedler semantics, as explained in Section \ref{section:sweedler_sem} below. For this reason we refer to $d_V$ as the \emph{dereliction map}. If there is no chance of confusion we will often write $d$ for $d_V$. Note that the comultiplication $\Delta$ on ${!} V$ and the counit $\varepsilon$ are derived from the comultiplication and counit on the summands $\Sym_P(V)$ in the obvious fashion; see \cite[p.50]{sweedler}. 

In the remainder of this section $V$ is an arbitrary vector space.

\begin{definition} For an integer $n > 0$ set $[n] = \{1,\ldots,n\}$.
\end{definition}

\begin{definition}[(Ket notation)]\label{defn:ket} For $P,v_1,\ldots,v_s \in V$ we write
\[
\ket{v_1,\ldots,v_s}_P \in \Sym_P(V) \subseteq {!} V
\]
for the image in the summand $\Sym_P(V)$ of the tensor $v_1 \otimes \cdots \otimes v_s \in V^{\otimes s}$. The identity $1 \in \Sym_P(V)$ is denoted $\ket{\emptyset}_P$. Given $I = \{i_1,\ldots,i_t\} \subseteq [s]$ we write
\be\label{eq:accum_ket}
\ket{v_I}_P := \ket{v_{i_1},\ldots,v_{i_t}}_P\,.
\ee
\end{definition}

In this notation the formulas for the comultiplication and counit are
\begin{align}
\Delta &: {!} V \lto {!} V \otimes {!} V\,, & \Delta \ket{v_1,\ldots,v_s}_P &= \sum_{I \subseteq [s]} \ket{v_I}_P \otimes \ket{v_{I^c}}_P\,,\label{eq:ket1}\\
\varepsilon &: {!} V \lto k\,, & \varepsilon \ket{v_1,\ldots,v_s}_P &= \delta_{s = 0} \cdot 1\label{eq:ket2}
\end{align}
where in the formula for $\Delta$, $I^c$ denotes the complement in $[s]$ and $I$ ranges over all subsets, including the empty set. In particular $\Delta \ket{\emptyset}_P = \ket{\emptyset}_P \otimes \ket{\emptyset}_P$ and $\varepsilon \ket{\emptyset}_P = 1$.

\subsection{Group-like elements are points}\label{section:group_like_elements}

In algebraic geometry a $k$-point of an algebra $A$ is by definition a morphism of algebras $A \lto k$, or equivalently a morphism of schemes $\Spec(k) \lto \Spec(A)$. In the case where $A = \Sym(V^*)$ for $V$ finite-dimensional such points are canonically identified with $V$ itself:
\begin{align*}
\Alg_k( \Sym(V^*), k ) \cong \Hom_k( V^*, k ) \cong V^{**} \cong V
\end{align*}
This bijection identifies a point $P = (P_1,\ldots,P_n)$ in $V$ with coordinates $P_i \in k$ in a chosen basis $e_1,\ldots,e_n$ for $V$ with the maximal ideal
\[
\mf{m}_P = (x_1 - P_1,\ldots,x_n-P_n) \subseteq \Sym(V^*)
\]
where as above we write $x_i = e_i^*$ for the dual basis. From the coalgebraic point of view a $k$-point of a coalgebra $C$ is a morphism of coalgebras $k \lto C$. Such morphisms are in canonical bijection with the \emph{group-like elements} \cite[p.57]{sweedler} of $C$
\[
G(C) = \{ x \in C \l \Delta(x) = x \otimes x \text{ and } \varepsilon(x) = 1 \}
\]
via the bijection
\begin{gather*}
\Coalg_k(k, C) \lto G(C)\\
\varphi \longmapsto \varphi(1)\,.
\end{gather*}
The points of a $k$-algebra $A$ are related to the points of $A^{\circ}$ since by Theorem \ref{theorem:fund_adjunc}
\begin{gather*}
\Alg_k( A, k ) \cong \Coalg_k(k, A^{\circ})
\end{gather*}
and so in particular
\[
V \cong \Alg_k( \Sym(V^*), k ) \cong \Coalg_k( k, \Sym(V^*)^{\circ} ) \cong \Coalg_k( k, {!} V ) \cong G( {!} V )\,.
\]
 
\subsection{Primitive elements are tangent vectors}\label{section:tangent_vectors}

From the point of view of algebraic geometry a tangent vector at a $k$-point of an algebra $A$ is a morphism of algebras $\varphi: A \lto k[\varepsilon]/(\varepsilon^2)$. The point of $A$ at which such a tangent vector is ``attached'' is given by the composite
\be\label{eq:point_of_tangent}
\xymatrix@C+2pc{
A \ar[r]^-{\varphi} & k[\varepsilon]/(\varepsilon^2) \ar[r] & k
}
\ee
where in the second map, $\varepsilon \longmapsto 0$. The appropriateness of this definition can be seen readily in the case $A = \Sym(V^*)$ for $V$ finite-dimensional where, using the coordinates $x_i$ of the previous section to identify $A$ with $k[x_1,\ldots,x_n]$, tangent vectors
\[
v = \sum_{i=1}^n v_i \frac{\partial}{\partial x_i} \qquad v_i \in k
\]
at a point $P = (P_1,\ldots,P_n)$ are in bijective correspondence with $k$-algebra morphisms
\begin{gather*}
\varphi_v: k[x_1,\ldots,x_n] \lto k[\varepsilon]/(\varepsilon^2)\\
f(x_1,\ldots,x_n) \longmapsto f(P) \cdot 1 + \sum_{i=1}^n v_i \frac{\partial f}{\partial x_i}\Big\vert_P \cdot \varepsilon\,.
\end{gather*}
From the coalgebraic point of view a tangent vector at a $k$-point of a coalgebra $C$ is a morphism of coalgebras $\varphi: (k[\varepsilon]/(\varepsilon^2))^* \lto C$. We set $\cat{T} = (k[\varepsilon]/(\varepsilon^2))^*$. The $k$-point at which the tangent vector is attached is given by the composite
\[
\xymatrix@C+2pc{
k \ar[r] & \cat{T} \ar[r]^-{\varphi} & C
}
\]
where the first map is the dual of the second map in \eqref{eq:point_of_tangent}. Such morphisms are in canonical bijection with the \emph{primitive elements} of $C$ \cite[p.199]{sweedler}
\[
\operatorname{Prim}(C) = \{ x \in C \l \Delta(x) = x \otimes g + g \otimes x \text{ for some } g \in G(C) \}
\]
via the bijection
\begin{gather*}
\Coalg_k( \cat{T}, C ) \lto \operatorname{Prim}(C)\\
\varphi \longmapsto \varphi(\varepsilon^*)\,.
\end{gather*}
The tangent vectors at points of a $k$-algebra $A$ are related to tangent vectors at points of the coalgebra $A^{\circ}$ via Theorem \ref{theorem:fund_adjunc}
\[
\Alg_k( A, k[\varepsilon]/(\varepsilon^2)) \cong \Alg_k( A, (k[\varepsilon]/(\varepsilon^2))^{**}) \cong \Coalg_k( \cat{T}, A^{\circ})
\]
and in particular there is a canonical bijection
\begin{align*}
\operatorname{Prim}({!} V) &\cong \Coalg_k(\cat{T}, {!} V)\\
&\cong \Coalg_k( \cat{T}, \Sym(V^*)^{\circ} )\\
&\cong \Alg_k( \Sym(V^*), k[\varepsilon]/(\varepsilon^2) )\\
&\cong \Hom_k( V^*, k[\varepsilon]/(\varepsilon^2) )\\
&\cong V \otimes ( k \oplus k \varepsilon )\\
&\cong V \oplus V \varepsilon\,.
\end{align*}
Given $P, v \in V$ the morphism of coalgebras $\cat{T} \lto {!} V$ corresponding to the pair $(P, v \varepsilon)$ under this bijection is precisely the morphism \eqref{eq:toucan} alluded to in the introduction.

\subsection{Definition of the Sweedler semantics}\label{section:sweedler_sem}

With the notation for the cofree coalgebra ${!} V$ and universal map $d_V$ as introduced above (see Definition \ref{defn:cofree_coalg_gen}) we now recall the definition of the Sweedler semantics $\den{-}$ in the category $\cat{V}$ of $k$-vector spaces from \cite{hyland} and \cite[\S 5.1, \S 5.3]{murfet_ll}. For each atomic formula $x$ of the logic, we choose a vector space $\dntn{x}$. For formulas $A, B$, define:
\begin{itemize}
\item $\dntn{A \otimes B} = \dntn{A} \otimes \dntn{B}$,
\item $\dntn{A \multimap B} = \Hom_k(\dntn{A}, \dntn{B})$,
\item $\dntn{{!} A} = {!}\dntn{A}$\,.
\end{itemize}
If $\Gamma$ is $A_1, ..., A_n$, then we define $\dntn{\Gamma} = \dntn{A_1} \otimes ... \otimes \dntn{A_n}$. Given a linear map $\pi: C \lto V$ with $C$ a coalgebra we write $\prom(\pi): C \lto {!} V$ for the unique morphism of coalgebras with $d_V \circ \prom(\pi) = \pi$ and we use the same notation in the syntax for the proof obtained from a proof $\pi$ of ${!}\Gamma \vdash B$ by applying the promotion rule to obtain a proof of ${!} \Gamma \vdash {!} B$.


\begin{definition}\label{defn: denotation of a proof}
The \emph{denotation} $\dntn{\pi}$ of a proof $\pi: \Gamma \vdash B$ is a linear map $\dntn{\pi}: \dntn{\Gamma} \to \dntn{B}$ defined recursively on the structure of proofs. The proof $\pi$ must match one of the proofs in the first column of the following table, and the second column gives its denotation. Here, $d$ denotes the dereliction map, $\Delta$ the comultiplication and $\varepsilon$ the counit:
\begin{center}
{{\tabulinesep=1.60mm
\begin{longtabu}{| c | c |}

\hline

\AxiomC{}
\RightLabel{\scriptsize axiom}
\UnaryInfC{$A \vdash A$}
\DisplayProof

&

$\dntn{\pi}(a) = a$

\\ \hline

\AxiomC{$\proofvdots{\pi_1}$}
\noLine\UnaryInfC{$\Gamma, A, B, \Delta \vdash C$}
\RightLabel{\scriptsize exch}
\UnaryInfC{$\Gamma, B, A, \Delta \vdash C$}
\DisplayProof

&

$\dntn{\pi}(\gamma \otimes b \otimes a \otimes \delta) = \dntn{\pi_1}(\gamma \otimes a \otimes b \otimes \delta)$

\\ \hline

\AxiomC{$\proofvdots{\pi_1}$}
\noLine\UnaryInfC{$\Gamma \vdash A$}
\AxiomC{$\proofvdots{\pi_2}$}
\noLine\UnaryInfC{$\Delta, A \vdash B$}
\RightLabel{\scriptsize cut}
\BinaryInfC{$\Gamma, \Delta \vdash B$}
\DisplayProof

&

$\dntn{\pi}(\gamma \otimes \delta) = \dntn{\pi_2}(\delta \otimes \dntn{\pi_1}(\gamma))$

\\ \hline

\AxiomC{$\proofvdots{\pi_1}$}
\noLine\UnaryInfC{$\Gamma, A, B \vdash C$}
\RightLabel{\scriptsize $\otimes L$}
\UnaryInfC{$\Gamma, A \otimes B \vdash C$}
\DisplayProof

&

$\dntn{\pi}(\gamma \otimes (a \otimes b)) = \dntn{\pi_1}(\gamma \otimes a \otimes b)$

\\ \hline

\AxiomC{$\proofvdots{\pi_1}$}
\noLine\UnaryInfC{$\Gamma \vdash A$}
\AxiomC{$\proofvdots{\pi_2}$}
\noLine\UnaryInfC{$\Delta \vdash B$}
\RightLabel{\scriptsize $\otimes R$}
\BinaryInfC{$\Gamma, \Delta \vdash A \otimes B$}
\DisplayProof

&  

$\dntn{\pi}(\gamma \otimes \delta) = \dntn{\pi_1}(\gamma) \otimes \dntn{\pi_2}(\delta)$

\\ \hline

\AxiomC{$\proofvdots{\pi_1}$}
\noLine\UnaryInfC{$\Gamma \vdash A$}
\AxiomC{$\proofvdots{\pi_2}$}
\noLine\UnaryInfC{$\Delta, B \vdash C$}
\RightLabel{\scriptsize $\multimap L$}
\BinaryInfC{$\Gamma, \Delta, A \multimap B \vdash C$}
\DisplayProof

&

$\dntn{\pi}(\gamma \otimes \delta \otimes \varphi) = \dntn{\pi_2}(\delta \otimes \varphi\circ\dntn{\pi_1}(\gamma))$

\\ \hline

\AxiomC{$\proofvdots{\pi_1}$}
\noLine\UnaryInfC{$\Gamma, A \vdash B$}
\RightLabel{\scriptsize $\multimap R$}
\UnaryInfC{$\Gamma \vdash A \multimap B$}
\DisplayProof

&

$\dntn{\pi}(\gamma) = \{a \mapsto \dntn{\pi_1}(\gamma \otimes a)\}$

\\ \hline

\AxiomC{$\proofvdots{\pi_1}$}
\noLine\UnaryInfC{$\Gamma, A \vdash B$}
\RightLabel{\scriptsize der}
\UnaryInfC{$\Gamma, {!}A \vdash B$}
\DisplayProof

&

$\dntn{\pi}(\gamma \otimes \overline{a}) = \dntn{\pi_1}(\gamma \otimes d(\overline{a}))$

\\ \hline

\AxiomC{$\proofvdots{\pi_1}$}
\noLine\UnaryInfC{${!}\Gamma \vdash A$}
\RightLabel{\scriptsize prom}
\UnaryInfC{${!}\Gamma \vdash {!}A$}
\DisplayProof

&

$\dntn{\pi} = \prom\dntn{\pi_1}$

\\ \hline

\AxiomC{$\proofvdots{\pi_1}$}
\noLine\UnaryInfC{$\Gamma \vdash B$}
\RightLabel{\scriptsize weak}
\UnaryInfC{$\Gamma, {!}A \vdash B$}
\DisplayProof

& 

$\dntn{\pi}(\gamma \otimes \overline{a}) = \varepsilon(\overline{a}) \dntn{\pi_1}(\gamma)$

\\ \hline

\AxiomC{$\proofvdots{\pi_1}$}
\noLine\UnaryInfC{$\Gamma, {!}A, {!}A \vdash B$}
\RightLabel{\scriptsize ctr}
\UnaryInfC{$\Gamma, {!}A \vdash B$}
\DisplayProof

& 

$\dntn{\pi}(\gamma \otimes \overline{a}) = \dntn{\pi_1}(\gamma \otimes \Delta(\overline{a}))$

\\ \hline
\end{longtabu}}
Table 2.1: Denotations of proofs in the Sweedler semantics.}
\end{center}
\end{definition}

\begin{definition}\label{defn:delta_map} For a vector space $V$ we denote by
\[
\delta_V: {!} V \lto {!} {!} V
\]
the unique morphism of coalgebras satisfying $d_{{!} V} \circ \delta_V = 1_{{!} V}$. Where there is no possibility of confusion we write $\delta$ for $\delta_V$.
\end{definition}

Given $P, v_1,\ldots,v_s \in V$ we have by \cite[Theorem 2.22]{murfet_coalg}
\be
\delta \ket{v_1, \ldots, v_s}_P = \sum_{\{C_1,\ldots,C_l\} \in \cat{P}_{[s]}} \Big|\, \ket{v_{C_1}}_P, \ldots, \ket{v_{C_l}}_P \Big\rangle_Q\,
\ee
where $\cat{P}_{[s]}$ denotes the set of partitions of $[s] = \{1,\ldots,s\}$ and $Q = \ket{\emptyset}_P$. Our partitions do not contain the empty set. As a special case $\delta\ket{\emptyset}_P = \ket{\emptyset}_Q$. 

Given a linear map $f: V \lto W$ there is a unique morphism of coalgebras ${!} f: {!} V \lto {!} W$ with the property that $d_W \circ {!} f = f \circ d_V$ and this makes ${!}$ into a functor ${!}: \cat{V} \lto \cat{V}$. In fact this functor is a comonad on $\cat{V}$, when equipped with the natural transformations $\delta: {!} \lto {!!}$ and $d: {!} \lto \operatorname{id}_{\cat{V}}$ which are component-wise the morphisms $\delta_V$ and $d_V$ defined above. If we let $\Delta, \varepsilon$ denote the assignment of comultiplication and counit maps to all the vector spaces ${!} V$ then, by construction, we have in the sense of \cite[Definition 2.1]{blutecs}:

\begin{lemma}\label{lemma:modality} The tuple $({!}, \delta, d, \Delta, \varepsilon)$ is a coalgebra modality on $\cat{V}$.
\end{lemma}

\subsection{Local cohomology and distributions}\label{section:residues}

In this section we explain how the cofree coalgebra arises in algebraic geometry, since this connection gives a useful context for the differential structure of the Sweedler semantics; the contents will however not be used in the sequel. For a finite-dimensional vector space $V$ of dimension $n$ with $R = \Sym(V^*)$, one proves using local duality \cite[Theorem 2.6]{murfet_coalg} that there is an isomorphism
\be\label{eq:otherbangV}
\bigoplus_{P \in V} H^{n}_{P}(R, \Omega^n_{R/k}) \cong \Sym(V^*)^{\circ}
\ee
where $H^n_P$ denotes local cohomology at $P$ \cite{residuesduality}. This isomorphism is defined by sending a class $\tau$ in the local cohomology at $P$ to the functional $f \mapsto \Res_P(f \tau)$ where $\Res_P$ denotes the generalised residue and $f \tau$ the action by $R$ on local cohomology. The isomorphism \eqref{eq:otherbangV} arises from isomorphisms $H^{n}_{P}(R, \Omega^n_{R/k}) \cong \Sym_P(V)$ identifying the identity $\ket{\emptyset}_P$ in $\Sym_P(V)$ with the class of the meromorphic differential form \cite[Definition 2.9]{murfet_coalg}
\be
\left[ \frac{\ud x_1 \wedge \cdots \wedge \ud x_n}{(x_1-P_1), \ldots, (x_n-P_n)} \right] \in H^{n}_{P}(R, \Omega^n_{R/k})\,.
\ee
It is easy to see that 
\begin{equation}\label{eq:residue_differentiates}
\Res_P\Big( f \ket{v}_P \Big) = \partial_{v}( f )|_{x = P}\,,
\end{equation}
and more generally that \cite[Lemma 2.13]{murfet_coalg}
\begin{equation}\label{eq:residue_differentiates2}
\Res_P\Big( f \ket{v_1,\ldots,v_s}_P \Big) = \partial_{v_1} \cdots \partial_{v_s}( f )|_{x = P}\,.
\end{equation}
Thus we may identify elements of ${!} V$ with functionals on the space of polynomial functions, given by evaluating derivatives at points of $V$.

\begin{remark} When $k = \mathbb{C}$ with $V = \mathbb{C} v$ and $z = v^*$ the generator of $R = \mathbb{C}[z]$, this is nothing but the Cauchy integral formula since we have
\be
\ket{\emptyset}_P = \left[ \frac{\ud z }{z-P} \right], \qquad \ket{v}_P = \left[ \frac{\ud z }{(z-P)^2} \right]
\ee
and the Cauchy formula says
\[
f'(P) = \frac{1}{2 \pi i} \oint_\gamma \frac{f(z)}{(z-P)^2} \ud z\,.
\]
\end{remark}

\begin{remark}\label{remark:distr} When $k = \mathbb{R}$ this agrees with the analytic theory of distributions, since by \cite[Theorem 3.2.1]{friedlander} the $\mathbb{C}$-vector space of distributions on the real manifold $V$ supported at a point $P$ is spanned by the functions
\[
f \longmapsto \partial_{v_1} \cdots \partial_{v_s}(f)|_{x=P}
\]
as $s \ge 0$ and $v_1,\ldots,v_s$ varies over all sequences in $V$. So in this case we can identify the coalgebra ${!} V \otimes_{\mathbb{R}} \mathbb{C}$ with the space of distributions on $V$ with finite support.

In the semantics of differential linear logic defined using finiteness spaces \cite{ehrhard-finiteness} and convenient vector spaces \cite{blutecon} the space ${!} V$ is a closure of the linear span of Dirac distributions (in our notation, $\ket{\emptyset}_P$) on $V$. More precisely, if $V$ is a finite-dimensional convenient vector space then ${!} V$ consists of distributions of compact support. For example, see \cite[Theorem 5.7]{blutecs} for the limit defining the distribution $\ket{v}_0$ in our notation. There is a similar role for Dirac distributions in the Coherent Banach space semantics of linear logic in \cite[\S 3.2]{girard_banach}. 

It is interesting to note that functional programs extended with Dirac distributions have already been considered in the literature on automatic differentiation; see \cite{nilsson}. For an abstract categorical theory of distributions via monads, see \cite{kock}.
\end{remark}

\begin{remark}
Any cocommutative coalgebra is the direct limit of finite-dimensional coalgebras, and the category of finite-dimensional cocommutative coalgebras is isomorphic to the category of zero-dimensional schemes over $k$. This is taken as the starting point of one approach to noncommutative geometry which has been influential in the study of $A_\infty$-algebras, where one posits that an arbitrary coalgebra is the coalgebra of distributions on a ``noncommutative space'' \cite[p.15]{kontnc}, \cite{kontnc2, lebruyn}.
\end{remark}

\section{Differential linear logic}

Let $k$ be an algebraically closed field of characteristic zero and $\cat{V}$ the category of $k$-vector spaces. This is a model of linear logic (see Section \ref{section:sweedler_sem}) when equipped with the comonad ${!}$ arising from the cofree coalgebra. We now explain how to equip this category with the structure necessary to make it a model of differential linear logic, following \cite{blutecs}.

Given vector spaces $V,W$ we write $\sigma_{V,W}: V \otimes W \lto W \otimes V$ for the linear \emph{swap map} defined on tensors by $\sigma_{V,W}(x \otimes y) = y \otimes x$. By \cite[Proposition 2.6]{blutecs} to equip $\cat{V}$ with the coalgebra modality $({!}, \delta, d, \Delta, \varepsilon)$ of Lemma \ref{lemma:modality} as a differential category, we need to define a deriving transformation in the sense of \cite[Definition 2.5]{blutecs}. 

\begin{definition} A \emph{deriving transformation} for $(\cat{V}, {!}, \delta, d, \Delta, \varepsilon)$ is a family of linear maps
\[
D_V: {!} V \otimes V \lto {!} V
\]
defined for all $V \in \cat{V}$ and natural in $V$, satisfying the following properties:
\begin{itemize}
\item[(D.1)] $\varepsilon \circ D = 0$, that is,
\be
\xymatrix@C+2pc{
{!} V \otimes V \ar[r]^-{D_V} & {!} V \ar[r]^-{\varepsilon} & k
} = 0\,.
\ee
\item[(D.2)] $\Delta \circ D = (1 \otimes D) \circ (\Delta \otimes 1) + (D \otimes 1) \circ (1 \otimes \sigma) \circ (\Delta \otimes 1)$, that is,
\be
\xymatrix@C+2pc{
{!} V \otimes V \ar[r]^-{D_V} & {!} V \ar[r]^-{\Delta} & {!} V \otimes {!} V
}
\ee
is equal to the sum
\begin{gather*}
\xymatrix@C+2pc{{!} V \otimes V \ar[r]^-{\Delta \otimes 1} & {!} V \otimes {!} V \otimes V \ar[r]^-{1 \otimes D_V} & {!} V \otimes {!} V} \quad +\\
\xymatrix@C+2pc{{!} V \otimes V \ar[r]^-{\Delta \otimes 1} & {!} V \otimes {!} V \otimes V \ar[r]^-{1 \otimes \sigma_{{!}V, V}}_{\cong} & {!} V \otimes V \otimes {!} V \ar[r]^-{D_V \otimes 1} & {!} V \otimes {!} V}
\end{gather*}
\item[(D.3)] $d \circ D = a \circ (\varepsilon \otimes 1)$, that is,
\be
\xymatrix@C+2pc{
{!} V \otimes V \ar[r]^-{D_V} & {!} V \ar[r]^-{d} & V 
}
\quad = \quad
\xymatrix@C+2pc{
{!} V \otimes V \ar[r]^-{\varepsilon \otimes 1} & k \otimes V \ar[r]^-{a}_-{\cong} & V
}
\ee
where $a(\lambda \otimes x) = \lambda x$.
\item[(D.4)] $\delta \circ D = D \circ (\delta \otimes D) \circ (\Delta \otimes 1)$, that is,
\be
\xymatrix@C+2pc{
{!} V \otimes V \ar[r]^-{D_V} & {!} V \ar[r]^-{\delta} & {!}{!} V
}
\ee
is equal to
\be\label{eq:D4final}
\xymatrix@C+2pc{
{!} V \otimes V \ar[r]^-{\Delta \otimes 1} & {!} V \otimes {!} V \otimes V \ar[r]^-{\delta \otimes D_V} & {!}{!} V \otimes {!} V \ar[r]^-{D_{{!}V}} & {!}{!} V\,.
}
\ee
\end{itemize}
\end{definition}

We refer to \cite[\S 2.2]{blutecs} for an explanation of these axioms. Briefly, (D.1) says the derivative of constant maps is zero, (D.2) is the product rule, (D.3) says the derivative of a linear map is constant, and (D.4) is the chain rule. Clearly the rules specify how to commute $D$ past the structural maps $\delta, d, \Delta, \varepsilon$. Here $d$ stands for the dereliction rule in linear logic, $\Delta$ for contraction and $\varepsilon$ for weakening. The map $\delta$ stands for promotion, since for a linear map $\phi: {!} V \lto W$ the unique lifting to a morphism of coalgebras $\Phi: {!} V \lto {!} W$ can be obtained as the composite
\be
\xymatrix@C+2pc{
{!} V \ar[r]^-{\delta} & {!!} V \ar[r]^-{{!} \phi} & {!} W\,.
}
\ee

\begin{definition}\label{defn:D} We define the $k$-linear map $D_V: {!} V \otimes V \lto {!}V$ by
\be\label{defn:DV}
D_V\big( \ket{v_1,\ldots,v_s}_P \otimes v \big) = \ket{ v, v_1,\ldots,v_s }_P\,.
\ee
\end{definition}

\begin{theorem}\label{main_theorem} $D_V$ is a deriving transformation for any vector space $V$.
\end{theorem}

We split the proof of the theorem into lemmas. We prefer to give the proofs without first choosing a basis of $V$, but if one is willing to do so, then the connection between these identities and the usual rules of calculus follows from writing the formula for the comultiplication $\Delta$ as a kind of Taylor expansion; see for example \cite[(B.65)]{seiler}. Throughout we use the ket notation of Definition \ref{defn:ket}.

\begin{lemma} (D.1) holds for $V$.
\end{lemma}
\begin{proof}
This is clear, since the counit $\varepsilon: {!} V \lto k$ vanishes on $\ket{v_1,\ldots,v_s}_P$ if $s > 0$.
\end{proof}

\begin{lemma} (D.2) holds for $V$.
\end{lemma}
\begin{proof}
Setting $v_0 = v$ we have (recall the notational conventions of Definition \ref{defn:ket})
\begin{align*}
\Delta D \big( \ket{v_1,\ldots,v_s}_P \otimes v \big) &= \Delta \ket{v, v_1, \ldots, v_s }_P\\
&= \sum_{I \subseteq \{0,1,\ldots,s\}} \ket{v_I}_P \otimes \ket{v_{I^c}}_P\\
&= \sum_{0 \in I} \ket{v_I}_P \otimes \ket{v_{I^c}}_P + \sum_{0 \notin I} \ket{v_I}_P \otimes \ket{v_{I^c}}_P \\
&= \sum_{J \subseteq \{1,\ldots,s\}}\ket{v, v_J}_P \otimes \ket{v_{J^c}}_P + \sum_{J \subseteq \{1,\ldots,s\}} \ket{v_J}_P \otimes \ket{v, v_{J^c}}_P \\
&= \sum_{J \subseteq \{1,\ldots,s\}} \Big\{ D\big( \ket{v_J}_P \otimes v \big) \otimes \ket{v_{J^c}}_P + \ket{v_J}_P \otimes D\big( \ket{v_{J^c}}_P \otimes v \big) \Big\}
\end{align*}
as claimed, where for $I \subseteq \{0,1,\ldots,s\}$ we write $I^c$ for $\{0,\ldots,s\} \setminus I$ and for $J \subseteq \{1,\ldots,s\}$, we write $J^c$ for $\{1,\ldots,s\} \setminus J$.
\end{proof}

\begin{lemma} (D.3) holds for $V$.
\end{lemma}
\begin{proof}
We have
\[
d D\big( \ket{v_1,\ldots,v_s}_P \otimes v \big) = d\ket{v,v_1,\ldots,v_s}_P = \delta_{s=0} v
\]
while
\[
a( \varepsilon \otimes 1 )\big( \ket{v_1,\ldots,v_s}_P \otimes v \big) = \varepsilon\ket{v_1,\ldots,v_s}_P \cdot v = \delta_{s=0} v\,.
\]
\end{proof}

\begin{lemma} (D.4) holds for $V$.
\end{lemma}
\begin{proof}
The trivial case is, with $Q = \ket{\emptyset}_P$,
\[
\delta D( \ket{\emptyset}_P \otimes v ) = \delta\ket{v}_P = \big|\, \ket{v}_P \big\rangle_Q
\]
and on the other side
\begin{align*}
D_{{!} V}( \delta \otimes D )( \Delta \otimes 1)( \ket{\emptyset}_P \otimes v ) &= D_{{!} V}( \delta \otimes D)( \ket{\emptyset}_P \otimes \ket{\emptyset}_P \otimes v )\\
&= D_{{!} V}( \ket{\emptyset}_Q \otimes \ket{v}_P )\\
&= \big|\, \ket{v}_P \big\rangle_Q\,.
\end{align*}
Now we consider the case $s > 0$. Putting $v_0 = v$ we have
\[
\delta D\big(  \ket{v_1,\ldots,v_s}_P \otimes v \big) = \sum_{X \in \cat{P}_{\{0,1,\ldots,s\}}} \Big| \prod_{x \in X} \ket{v_{x}}_P \Big\rangle_Q
\]
where for a partition $X = \{ x_1,\ldots,x_t \}$ the notation means multiplication in the symmetric algebra, that is
\[
\Big| \prod_{x \in X} \ket{v_{x}}_P \Big\rangle_Q = \Big|\, \ket{v_{x_1}}_P, \ldots, \ket{v_{x_t}}_P \Big\rangle_Q\,. 
\]
There is a surjective function
\begin{gather*}
\theta: \cat{P}_{\{0,1,\ldots,s\}} \lto \cat{P}_{\{1,\ldots,s\}}\\
\theta( X ) = \big\{ x \setminus \{0\} \l x \in X \text{ and } x \neq \{0\} \big\}
\end{gather*}
and given a partition $X = \{ x_1,\ldots,x_t \}$ of $\{1,\ldots,s\}$,
\begin{align*}
\theta^{-1}(X) &= \Big\{ \{ x_1 \cup \{0\}, x_2, \ldots, x_t \},\\
&\quad\{ x_1, x_2 \cup \{0\}, \ldots, x_t \},\\
&\quad\ldots,\\
&\quad\{ x_1, x_2, \ldots, x_{t-1}, x_t \cup \{0\} \}\\
&\quad\{ x_1, x_2, \ldots, x_t, \{0\} \} \Big\}\,.
\end{align*}
With this in mind we have, writing $\prod_{x' \neq x} \ket{v_{x'}}_P$ for the product in the symmetric algebra of $\ket{v_{x'}}_P$ as $x'$ ranges over elements of $X \setminus \{x\}$, that $\delta D\big(  \ket{v_1,\ldots,v_s}_P \otimes v \big)$ is equal to
\be\label{eq_written72}
\sum_{X \in \cat{P}_{\{1,\ldots,s\}}} \Big\{ \sum_{x \in X} \Big|\, \ket{v, v_x}_P\,, \prod_{x' \neq x} \ket{v_{x'}}_P\Big\rangle_Q + \Big|\, \ket{v}_P\,, \prod_{x \in X} \ket{v_x}_P \Big\rangle_Q \Big\}\,.
\ee
Note that when $X = \big\{ \{ 1,\ldots,s \} \big\}$ the summand is
\[
\Big|\, \ket{v,v_1,\ldots,v_s}_P \Big\rangle_Q + \Big|\, \ket{v}_P, \ket{v_1,\ldots,v_s}_P \Big\rangle_Q\,.
\]
On the other hand, the right hand side \eqref{eq:D4final} of the (D.4) identity is
\begin{align*}
&D_{{!} V}(\delta \otimes D)(\Delta \otimes 1)\big(  \ket{v_1,\ldots,v_s}_P \otimes v \big)\\
&= \sum_{I \subseteq \{1,\ldots,s\}} D_{{!}V}(\delta \otimes D)\Big( \ket{v_I}_P \otimes \ket{v_{I^c}}_P \otimes v \Big)\\
&= \sum_{I \subseteq \{1,\ldots,s\}} D_{{!} V}\Big( \delta \ket{v_I}_P \otimes \ket{v, v_{I^c}}_P \Big)\\
&= D_{{!} V}\Big( \delta\ket{\emptyset}_P \otimes \ket{v, v_1, \ldots, v_s }_P \Big)\\
&\qquad + \sum_{\emptyset \subset I} \sum_{Y \in \cat{P}_I} D_{{!}V}\Big( \big| \prod_{y \in Y} \ket{v_y}_P \big\rangle_Q \otimes \ket{v, v_{I^c}}_P \Big)
\end{align*}
\begin{align*}
&= \Big|\, \ket{ v, v_1,\ldots,v_s }_P \Big\rangle_Q + \sum_{\emptyset \subset I}\sum_{Y \in \cat{P}_I}  \Big|\, \ket{v, v_{I^c}}_P\,, \prod_{y \in Y} \ket{v_y}_P \Big\rangle_Q\\
&= \Big|\, \ket{ v, v_1,\ldots,v_s }_P \Big\rangle_Q + \sum_{Y \in \cat{P}_{\{1,\ldots,s\}}} \Big|\, \ket{v}_P\,, \prod_{y \in Y} \ket{v_y}_P \Big\rangle_Q\\
&\qquad+ \sum_{\emptyset \subset I \subset \{1,\ldots,s\}} \sum_{Y \in \cat{P}_I} \Big|\, \ket{v, v_{I^c}}_P\,, \prod_{y \in Y} \ket{v_y}_P \Big\rangle_Q
\end{align*}
which matches \eqref{eq_written72} since the last sum can be rewritten as
\[
\sum_{\substack{X \in \cat{P}_{\{1,\ldots,s\}}\\ \text{ with } |X| > 1}} \sum_{x \in X} \Big|\, \ket{v, v_x}_P\,, \prod_{x' \neq x} \ket{v_{x'}}_P \Big\rangle_Q\,.
\]
\end{proof}

Together the previous lemmas complete the proof of Theorem \ref{main_theorem}.

\begin{corollary}\label{corollary:diffcat} $\cat{V}$ is a differential category.
\end{corollary}
\begin{proof}
This follows from \cite[Proposition 2.6]{blutecs}.
\end{proof}

We are grateful to the referee who pointed out to us (see \cite[\S4.1]{cite4}) that an additional axiom is now routinely added to deriving transformations:

\begin{lemma}
For any vector space $V$, the diagram
\[
\begin{tikzcd}[column sep=large]
{!}V \otimes V \otimes V \arrow[dd, swap,"{D_V \otimes 1}"]\arrow[r, "{1 \otimes \sigma_{V,V}}"]& {!}V \otimes V \otimes V \arrow[r, "{D_V \otimes 1}"]& {!}V \otimes V\arrow[dd, "{D_V}"] \\ \\
{!}V \otimes V \arrow[rr, swap, "{D_V}"]& & {!}V
\end{tikzcd}
\]
commutes.
\end{lemma}

\begin{proof}
We compute:
\begin{align*}
D_V(D_V \otimes 1)(1 \otimes \sigma)\left(\ket{v_1, ..., v_s}_P \otimes v \otimes v'\right)
&= D_V(D_V \otimes 1)\left(\ket{v_1, ..., v_s}_P \otimes v' \otimes v\right)
\\&= D_V\left(\ket{v', v_1, ..., v_s}_P \otimes v\right)
\\&= \ket{v, v', v_1, ..., v_s}_P
\\&= \ket{v', v, v_1, ..., v_s}_P
\\&= \ket{v, v_1, ..., v_s}_P \otimes v'
\\&= D_V\left(\ket{v, v_1, ..., v_s}_P \otimes v'\right)
\\&= D_V(D_V \otimes 1)\left(\ket{v_1, ..., v_s}_P \otimes v \otimes v'\right).
\end{align*}
\end{proof}

\begin{remark}\label{remark:justify}
From the geometric point of view (see Section \ref{section:residues}) the definition of $D_V$ is justified as follows: adding $v$ to a ket contributes, inside a residue, a partial derivative in the direction $v$ by \eqref{eq:residue_differentiates2}. To state this more formally, let $\cat{D}(R)$ denote the algebra of $k$-linear differential operators on $R = \Sym(V^*)$ and observe that there is a canonical map $\iota: V \lto \cat{D}(R)$ sending $v \in V$ to the differential operator $\partial_v$ and we have a map
\be
\xymatrix@C+2pc{
H^n_P(R, \Omega^n_{R/k}) \otimes V \ar[r]^-{1 \otimes \iota} & H^n_P(R, \Omega^n_{R/k}) \otimes \cat{D}(R) \ar[r]^-{a} & H^n_P(R, \Omega^n_{R/k})
}
\ee
where $a$ denotes the action of the ring $\cat{D}(R)$ on local cohomology \cite[Lemma 2.7]{murfet_coalg}. These maps assemble in the colimit \eqref{eq:otherbangV} to give \eqref{defn:DV}.
\end{remark}

\subsection{Codereliction, cocontraction, coweakening}\label{section:coder}

An alternative formulation of the differential structure in differential linear logic is in terms of \emph{codereliction, cocontraction} and \emph{coweakening} maps; see \cite{fiore} and \cite[\S 5.1]{blutecon}. This has the advantage of providing an appealing symmetry to the formulation of the syntax. In this section we briefly sketch the definition of these maps in the Sweedler semantics. 
\vspace{0.2cm}

First we recall the canonical commutative Hopf structure on ${!} V$ of \cite[\S 6.4]{sweedler}.

\begin{lemma}\label{lemma:bangisadditive} Given vector spaces $V_1,V_2$ then there is an isomorphism of coalgebras
\begin{gather*}
\Theta: {!} V_1 \otimes {!} V_2 \lto {!} (V_1 \oplus V_2)\,,\\
\ket{v_1,\ldots,v_s}_P \otimes \ket{w_1,\ldots,w_t}_Q \longmapsto \ket{v_1,\ldots,v_s,w_1,\ldots,w_t}_{(P,Q)}\,.
\end{gather*}
\end{lemma}
\begin{proof}
The existence of this isomorphism is due to Sweedler, for the explicit calculation of the map see \cite[Remark 2.19]{sweedler}.
\end{proof}

Using this and the definitions in \cite{sweedler} it is easy to check that the product $\nabla$ is
\begin{gather*}
\nabla: {!} V \otimes {!} V \lto {!} V\,,\\
\ket{v_1,\ldots,v_s}_P \otimes \ket{w_1,\ldots,w_t}_Q \longmapsto \ket{v_1,\ldots,v_s,w_1,\ldots,w_t}_{P+Q}\,,
\end{gather*}
while the antipode $S$ is
\begin{gather*}
S: {!} V \lto {!} V\,,\\
\ket{v_1,\ldots,v_s}_P \longmapsto \ket{-v_1,\ldots,-v_s}_{-P}
\end{gather*}
and the unit $u: k \lto {!} V$ is $u(1) = \ket{\emptyset}_0$. By \cite[Theorem 6.4.8]{sweedler} these maps make ${!} V$ into a commutative (and cocommutative) Hopf algebra. In the terminology of \cite{ehrhard-survey} the map $\nabla$ is the \emph{cocontraction} map and $u$ is the \emph{coweakening} map (the antipode seems not to have a formal role in differential linear logic). Finally,

\begin{definition} The \emph{codereliction} $\bar{d}$ is the composite
\[
\xymatrix@C+2pc{
V \cong k \otimes V \ar[r]^-{u \otimes 1} & {!} V \otimes V \ar[r]^-{D} & {!} V
}
\]
which is given by $v \mapsto \ket{v}_0$.
\end{definition}

Note that we can recover $D$ as
\begin{gather*}
\xymatrix@C+2pc{
{!} V \otimes V \ar[r]^-{1 \otimes \bar{d}} & {!} V \otimes {!} V \ar[r]^-{\nabla} & {!} V
}\\
\ket{v_1,\ldots,v_s}_P \otimes v \mapsto \ket{v_1,\ldots,v_s}_P \otimes \ket{v}_0 \mapsto \ket{v,v_1,\ldots,v_s}_P\,.
\end{gather*}
It seems more convenient to model differentiation syntactically using the codereliction, cocontraction and coweakening maps, rather than the deriving transformation $D$ itself. We briefly sketch how this works, following \cite{ehrhard-survey}. In the sequent calculus for linear logic one introduces three new deduction rules ``dual'' to dereliction, contraction and weakening:
\[
\AxiomC{$\Gamma, {!}A, \Delta \vdash B$}
\LeftLabel{(Codereliction): }
\RightLabel{\scriptsize coder}
\UnaryInfC{$\Gamma, A, \Delta \vdash B$}
\DisplayProof
\]
\[
\AxiomC{$\Gamma, !A, \Delta \vdash B$}
\LeftLabel{(Cocontraction): }
\RightLabel{\scriptsize coctr}
\UnaryInfC{$\Gamma, !A, !A, \Delta \vdash B$}
\DisplayProof
\]
\[
\AxiomC{$\Gamma, !A, \Delta \vdash B$}
\LeftLabel{(Coweakening): }
\RightLabel{\scriptsize coweak}
\UnaryInfC{$\Gamma, \Delta \vdash B$}
\DisplayProof
\]
together with new cut-elimination rules \cite[\S 1.4.3]{ehrhard-survey}. 

\begin{definition}\label{defn:derivative_proof} Given a proof $\pi$ of ${!} A \vdash B$ in linear logic, the \emph{derivative} $\partial \pi$ is the proof
\be
\begin{mathprooftree}
\AxiomC{$\pi$}
\noLine\UnaryInfC{$\vdots$}
\def\extraVskip{5pt}
\noLine\UnaryInfC{${!} A \vdash B$}
\RightLabel{\scriptsize coctr}
\UnaryInfC{${!} A, {!} A \vdash B$}
\RightLabel{\scriptsize coder}
\UnaryInfC{${!} A, A \vdash B$}
\end{mathprooftree}
\ee
whose denotation is, by our earlier remark, the composite
\be
\xymatrix@C+2pc{
{!} \den{A} \otimes \den{A} \ar[r]^-{D} & {!} \den{A} \ar[r]^-{\den{\pi}} & \den{B}\,.
}
\ee
\end{definition}

\begin{remark}\label{remark:totem}
Given $\pi$ as above we have the function \cite[Definition 5.10]{murfet_ll}
\be
\den{\pi}_{nl}: \den{A} \lto \den{B}\,, \qquad P \longmapsto \den{\pi}\ket{\emptyset}_P\,,
\ee
and for $P, v \in \den{A}$ we interpret the vector
\be
\den{\pi}D( \ket{\emptyset}_P \otimes v ) = \den{\pi}\ket{v}_P \in \den{B}
\ee
as the derivative of $\den{\pi}_{nl}$ at the point $P$ in the direction $v$. Here we implicitly identify $\den{A}$ with the tangent space $T_P\den{A}$ and $\den{B}$ with the tangent space $T_{\den{\pi}_{nl}(P)} \den{B}$. This interpretation is justified by the following elaboration of the remarks in the Introduction. 

Let $\operatorname{prom}(\pi)$ denote the proof which is the promotion of $\pi$, which has for its denotation the unique morphism of coalgebras $\den{ \operatorname{prom}(\pi) }: {!} \den{A} \lto {!} \den{B}$ with $d \circ \den{ \operatorname{prom}(\pi) } = \den{\pi}$. Let $\gamma: (k[\varepsilon]/\varepsilon^2)^* \lto {!} \den{A}$ be the morphism of coalgebras as in Section \ref{section:tangent_vectors} corresponding to the tangent vector $v$ at a point $P \in \den{A}$. Then the morphism of coalgebras
\be\label{eq:prompiafterpsi}
\den{ \operatorname{prom}(\pi) } \circ \gamma : (k[\varepsilon]/\varepsilon^2)^* \lto {!} \den{B}
\ee
has the following values, writing $Q = \den{\pi}_{nl}(P)$, we have by \cite[Theorem 2.22]{murfet_coalg}
\begin{align*}
\den{ \operatorname{prom}(\pi) }\gamma(1) &= \den{ \operatorname{prom}(\pi) } \ket{\emptyset}_P = \ket{\emptyset}_Q\,,\\
\den{ \operatorname{prom}(\pi) }\gamma(\varepsilon^*) &= \den{ \operatorname{prom}(\pi) } \ket{v}_P = \Big|\, \den{\pi}\ket{v}_P \Big\rangle_Q\,.
\end{align*}
Under the bijection of Section \ref{section:tangent_vectors} the morphism of coalgebras \eqref{eq:prompiafterpsi} therefore corresponds to the tangent vector $\den{\pi}\ket{v}_P \in \den{B}$ at $Q$.
\end{remark}

It is easy using the formulas for $\nabla, D$ to check that the $\nabla$-rule of \cite[\S 4.3]{blutecs} is satisfied:

\begin{lemma}\label{lemma_nablarule} The diagram
\be
\xymatrix@C+2pc@R+1.5pc{
V \otimes {!} V \otimes {!} V \ar[d]_-{\sigma_{V,{!}V} \otimes 1} \ar[r]^-{1 \otimes \nabla} & V \otimes {!} V \ar[r]^-{\sigma_{V,{!}V}} & {!} V \otimes V \ar[d]^-{D}\\
{!} V \otimes V \otimes {!} V \ar[r]_-{D \otimes 1} & {!} V \otimes {!} V \ar[r]_-{\nabla} & {!} V
}
\ee
commutes.
\end{lemma}

This, together with \cite[Theorem 4.12]{blutecs}, shows that $\cat{V}$ with the comonad ${!}$ and deriving transformation $D$ is a model of the differential calculus in the sense of \cite[Definition 4.11]{blutecs}.

\begin{corollary} $\cat{V}$ is a categorical model of the differential calculus.
\end{corollary}

\section{Examples}\label{section:examples}

In this section we give various examples of proofs $\pi$ and the derivatives $\den{\pi} \circ D$ of their denotations, according to Definition \ref{defn:derivative_proof}. The encoding of integers and binary sequences in linear logic is based on the following encoding of the composition rule.

\begin{definition} For any formula $A$ let $C^1_A$ denote the proof
\begin{center}
\AxiomC{}
\UnaryInfC{$A \vdash A$}
\AxiomC{}
\UnaryInfC{$A \vdash A$}
\RightLabel{\scriptsize$\multimap L$}
\BinaryInfC{$A, A \multimap A \vdash A$}
\DisplayProof
\end{center}
We define recursively for $n > 1$ a proof $C^n_A$ of $A, (A \multimap A)^n \vdash A$, where $(A \multimap A)^n$ denotes a sequence of $n$ copies of $A \multimap A$, to be
\begin{center}
\AxiomC{}
\UnaryInfC{$A \vdash A$}
\AxiomC{$C^{n-1}_A$}
\noLine\UnaryInfC{$\vdots$}
\noLine\UnaryInfC{$A, (A \multimap A)^{n-1} \vdash A$}
\RightLabel{\scriptsize$\multimap L$}
\BinaryInfC{$A, (A \multimap A)^n \vdash A$}
\DisplayProof
\end{center}
\end{definition} 

\begin{definition} For $n \ge 1$ let $\comp^n_A$ denote the proof
\begin{center}
\AxiomC{$C^{n}_A$}
\noLine\UnaryInfC{$\vdots$}
\noLine\UnaryInfC{$A, (A \multimap A)^{n} \vdash A$}
\RightLabel{\scriptsize$\multimap R$}
\UnaryInfC{$(A \multimap A)^n \vdash A \multimap A$}
\DisplayProof
\end{center}
We define $\comp^0_A$ to be the proof
\begin{center}
\AxiomC{}
\UnaryInfC{$A \vdash A$}
\RightLabel{\scriptsize$\multimap R$}
\UnaryInfC{$\vdash A \multimap A$}
\DisplayProof
\end{center}
\end{definition}

\begin{remark}
If $V = \den{A}$ and $\alpha_i \in \den{A \multimap A} = \End_k(V)$ for $1 \le i \le n$ then
\be\label{order_comp}
\den{\comp^n_A}( \alpha_1 \otimes \cdots \otimes \alpha_n ) = \alpha_n \circ \cdots \circ \alpha_1\,,
\ee
while $\den{\comp^0_A} = 1_V$. Note the reversed ordering on the right hand side!
\end{remark}

\subsection{Church numerals}\label{section:church}

\begin{definition} The type of \emph{integers on $A$} \cite[\S 5.3.2]{girard_llogic} is:
\[
\inta_A = {!}( A \multimap A ) \multimap (A \multimap A)\,.
\]
For $n \ge 0$ we define the Church numeral $\underline{n}_A$ to be the proof
\begin{center}
\AxiomC{$\comp^n_A$}
\noLine\UnaryInfC{$\vdots$}
\noLine\UnaryInfC{$(A \multimap A)^{n} \vdash A \multimap A$}
\doubleLine\RightLabel{\scriptsize $n \times$ der}
\UnaryInfC{${!}(A \multimap A)^n \vdash A \multimap A$}
\doubleLine\RightLabel{\scriptsize $n \times$ ctr}
\UnaryInfC{${!}(A \multimap A) \vdash A \multimap A$}
\RightLabel{\scriptsize $\multimap R$}
\UnaryInfC{$\vdash \inta_A$}
\DisplayProof
\end{center}
Generally we omit the final step, since it is irrelevant semantically. In the case $n = 0$ the formula ${!}(A \multimap A)$ is introduced on the left by a weakening rule.
\end{definition}

\begin{example} The proof $\underline{2}_A$ (see e.g. \cite[Example 5.9]{murfet_ll}) is
\begin{center}
\AxiomC{}
\UnaryInfC{$A \vdash A$}
\AxiomC{}
\UnaryInfC{$A \vdash A$}
\AxiomC{}
\UnaryInfC{$A \vdash A$}
\RightLabel{\scriptsize$\multimap L$}
\BinaryInfC{$A, A \multimap A \vdash A$}
\RightLabel{\scriptsize$\multimap L$}
\BinaryInfC{$A, A \multimap A, A \multimap A \vdash A$}
\RightLabel{\scriptsize$\multimap R$}
\UnaryInfC{$A \multimap A, A \multimap A \vdash A \multimap A$}
\RightLabel{\scriptsize der}
\UnaryInfC{$!( A \multimap A ), A \multimap A \vdash A \multimap A$}
\RightLabel{\scriptsize der}
\UnaryInfC{$!( A \multimap A), !(A \multimap A) \vdash A \multimap A$}
\RightLabel{\scriptsize ctr}
\UnaryInfC{$!( A \multimap A) \vdash A \multimap A$}
\DisplayProof
\end{center}
\end{example}

From now on $A$ is fixed and we write $\underline{n}$ for $\underline{n}_A$. Let $V = \den{A}$ so $\den{A \multimap A} = \End_k(V)$. In the notation of Remark \ref{remark:totem}, there is a function
\be
\den{\underline{n}}_{nl}: \End_k(V) \lto \End_k(V)\,.
\ee

\begin{lemma} For $n \ge 0$ and $\alpha \in \End_k(V)$, we have $\den{\underline{n}}\ket{\emptyset}_\alpha = \alpha^n$ so $\den{\underline{n}}_{nl}(\alpha) = \alpha^n$.
\end{lemma}
\begin{proof}
This is an easy exercise, see \cite{murfet_ll} for the case $n = 2$.
\end{proof}

The derivative $\partial\, \underline{n}$ of Definition \ref{defn:derivative_proof} is a proof of ${!}(A \multimap A), A \multimap A \vdash A \multimap A$ and for $\alpha, v \in \End_k(V)$ the value of its denotation $\den{\partial\, \underline{n}} = \den{\underline{n}} \circ D$ on $\ket{\emptyset}_\alpha \otimes v$, that is, the derivative of $\underline{n}$ at $\alpha$ in the direction of $v$, is $\den{\underline{n}} \ket{v}_\alpha$. 

\begin{lemma}\label{lemma:nderiv} $\den{\underline{n}}\ket{v}_\alpha = \sum_{i = 1}^{n} \alpha^{i-1} v \alpha^{n-i}$.
\end{lemma}
\begin{proof}
This may be computed using the formulas of \cite[p.19]{murfet_ll}. For example, in the case $n = 2$ the image of $\ket{v}_\alpha$ under $\den{\underline{n}}$ is given by
\begin{align*}
\ket{v}_\alpha 
&\xmapsto{\makebox[1cm]{\scriptsize ctr}} \ket{v}_\alpha \otimes \ket{\emptyset}_\alpha + \ket{\emptyset}_\alpha \otimes \ket{v}_\alpha
\\&\xmapsto{\makebox[1cm]{\scriptsize $2 \times$ \text{der}} } v \otimes \alpha + \alpha \otimes v
\\&\xmapsto{\makebox[1cm]{\scriptsize $ - \circ -$} } \alpha \circ v + v \circ \alpha\,,
\end{align*}
as claimed.
\end{proof}

\begin{remark} When $k = \mathbb{C}$, $V$ is $r$-dimensional and $\varphi = \den{\underline{n}}_{nl}$, the vector $\den{\underline{n}}\ket{v}_\alpha$ agrees with the image of $v$ under the usual tangent map of the smooth map $\varphi$
\[
\xymatrix@C+2pc{
M_r(\mathbb{C}) \cong T_\alpha \End_k(V) \ar[r]^-{T_\alpha \varphi} & T_{\alpha^n}\End_k(V) \cong M_r(\mathbb{C})\,.
}
\]
This justifies in this case the interpretation of $\den{\underline{n}}\ket{v}_\alpha$ as the derivative.
\end{remark}

\subsection{Binary integers}\label{section:bint}

\begin{definition}
The type of \emph{binary integers on $A$} \cite[\S 2.5.3]{girard_complexity} is:
\[
\binta_A = {!}( A \multimap A) \multimap ({!}( A \multimap A) \multimap ( A \multimap A)).
\]
Given a sequence $S \in \{0,1\}^*$ we define a proof $\underline{S}_A$ of $\binta_A$ as follows. Let $l \ge 0$ be the length of $S$. The proof tree for $\underline{S}_A$ matches that of the Church numeral $\underline{l}$ up to the step where we perform contractions, that is,
\begin{equation}\label{bint_Upto}
\begin{mathprooftree}
\AxiomC{$\comp^l_A$}
\noLine\UnaryInfC{$\vdots$}
\noLine\UnaryInfC{$(A \multimap A)^{l} \vdash A \multimap A$}
\doubleLine\RightLabel{\scriptsize $n \times$ der}
\UnaryInfC{${!}(A \multimap A)^l \vdash A \multimap A$}
\end{mathprooftree}
\end{equation}
We match each copy of ${!}(A \multimap A)$ on the left with the corresponding position in $S$, and using a series of contractions we identify all copies corresponding to a position in which $0$ appears in $S$, and likewise all copies corresponding to positions with a $1$. After these contractions, there will be two copies of ${!}(A \multimap A)$ on the left (the first being by convention the remnant of all the $0$-associated copies) unless $S$ contains only $0$'s or only $1$'s. In this case we use further a weakening rule to introduce the ``missing'' ${!}(A \multimap A)$, giving finally the desired proof $\underline{S}_A$:
\begin{center}
\AxiomC{$\comp^l_A$}
\noLine\UnaryInfC{$\vdots$}
\noLine\UnaryInfC{$(A \multimap A)^{l} \vdash A \multimap A$}
\doubleLine\RightLabel{\scriptsize $n \times$ der}
\UnaryInfC{${!}(A \multimap A)^l \vdash A \multimap A$}
\doubleLine\RightLabel{\scriptsize ctr and possibly weak}
\UnaryInfC{${!}(A \multimap A), {!}(A \multimap A) \vdash A \multimap A$}
\doubleLine\RightLabel{\scriptsize $2 \times \multimap R$}
\UnaryInfC{$\vdash \binta_A$}
\DisplayProof
\end{center}
In the final right $\multimap R$ introduction rules, the second copy of ${!}(A \multimap A)$ (associated with the $1$'s in $S$) is moved across the turnstile first. If $S$ is the empty sequence, then $l = 0$ and the proof is a pair of weakenings on the left followed by the $\multimap R$ introduction rules.
\end{definition}

For the rest of this section $A$ is fixed and we write $\underline{S}$ for $\underline{S}_A$.

\begin{example} The proof $\underline{001}$ is 
\begin{center}
\AxiomC{}
\UnaryInfC{$A \vdash A$}
\AxiomC{}
\UnaryInfC{$A \vdash A$}
\AxiomC{}
\UnaryInfC{$A \vdash A$}
\AxiomC{}
\UnaryInfC{$A \vdash A$}
\RightLabel{\scriptsize$\multimap L$}
\BinaryInfC{$A, A \multimap A \vdash A$}
\RightLabel{\scriptsize$\multimap L$}
\BinaryInfC{$A, A \multimap A, A \multimap A \vdash A$}
\RightLabel{\scriptsize$\multimap L$}
\BinaryInfC{$A, A \multimap A, A \multimap A, A \multimap A \vdash A$}
\RightLabel{\scriptsize$\multimap R$}
\UnaryInfC{$A \multimap A, A \multimap A, A \multimap A \vdash A \multimap A$}
\doubleLine\RightLabel{\scriptsize $3 \times$ der}
\UnaryInfC{$\textcolor{red}{!(A \multimap A)}, \textcolor{red}{!( A \multimap A)}, \textcolor{blue}{!(A \multimap A)} \vdash A \multimap A$}
\RightLabel{\scriptsize ctr}
\UnaryInfC{$\textcolor{red}{!(A \multimap A)}, \textcolor{blue}{!( A \multimap A)} \vdash A \multimap A$}
\doubleLine\RightLabel{\scriptsize $2 \times \multimap R$}
\UnaryInfC{$\vdash \binta_A$}
\DisplayProof
\end{center}
where the colouring indicates which copies of ${!}(A \multimap A)$ are contracted. Using \eqref{order_comp},
\be
\den{\underline{001}}\big(\vacu_\gamma \otimes \vacu_\delta\big) = \den{\comp^3_A}\big(\vacu_\gamma \otimes \vacu_\gamma \otimes \vacu_\delta \big) = \delta \circ \gamma \circ \gamma\,.
\ee
\end{example}

Generalising the calculation of Section \ref{section:church} we now describe the derivatives of binary integers. The general formula computes, for $S \in \{0,1\}^*$, the linear operator
\[
\den{ \underline{S} }\big( \ket{\alpha_1,\ldots,\alpha_r}_{\gamma} \otimes \ket{\beta_1,\ldots,\beta_s}_{\delta} \big) \in \End_k(V)\,.
\]
Informally, this operator is described by inserting $\gamma$ for $0$ and $\delta$ for $1$ in (the reversal of) $S$, and then summing over all ways of replacing $r$ of the $\gamma$'s in this composite with $\alpha_i$'s, and $t$ of the $\delta$'s with $\beta_j$'s. Let $\operatorname{Inj}(P,Q)$ denote the set of injective functions $P \lto Q$, and write $[s] = \{1,\ldots,s\}$.

\begin{lemma}\label{lemma:derivative_bint} Let $S = a_l a_{l-1} \cdots a_1$ with $a_i \in \{0,1\}$ be a binary sequence, and set
\[
N_0 = \{ j \l a_j = 0 \}\,, \qquad N_1 = \{ j \l a_j = 1 \}\,.
\]
Then we have
\be\label{eq:formula_derivative_bint}
\den{ \underline{S} }\big( \ket{\alpha_1,\ldots,\alpha_s}_{\gamma} \otimes \ket{\beta_1,\ldots,\beta_r}_{\delta} \big) = \sum_{f \in \operatorname{Inj}([s],N_0)} \sum_{g \in \operatorname{Inj}([r],N_1)} \Gamma^{f,g}_1 \circ \cdots \circ \Gamma^{f,g}_l\,,
\ee
where
\[
\Gamma^{f,g}_i = \begin{cases}
\gamma & i \in N_0 \setminus \operatorname{Im}(f)\,,\\
\delta & i \in N_1 \setminus \operatorname{Im}(g)\,,\\
\alpha_j & \text{if } i \in \operatorname{Im}(f) \text{ and } f(j) = i\,,\\
\beta_j & \text{if } i \in \operatorname{Im}(g) \text{ and } g(j) = i\,.
\end{cases}
\]
In particular this vanishes if $s > |N_0|$ or $r > |N_1|$.
\end{lemma}
\begin{proof}
Since $\den{\underline{S}}$ applies $n = |N_0|$ coproducts to $\ket{\alpha_1,\ldots,\alpha_s}_{\gamma}$ yielding
\[
\sum_{\substack{J_1,\ldots,J_n\\ \text{pairwise disjoint, s.t.} \\ J_1 \cup \cdots \cup J_n = \{1,\ldots,s\}}} \ket{\alpha_{J_1}}_\gamma \otimes \cdots \otimes \ket{\alpha_{J_n}}_\gamma\,,
\]
to which the dereliction operator $d^{\otimes n}$ is applied, which annihilates those tuples $(J_1,\ldots,J_n)$ where any $J_i$ contains more than one element. The resulting sum is over $f \in \operatorname{Int}([s],N_0)$ and each summand is $\gamma \otimes \cdots \otimes \alpha_{\sigma(1)} \otimes \cdots \otimes \gamma \otimes \cdots \alpha_{\sigma(s)} \otimes \cdots \otimes \gamma$ for a permutation $\sigma$. The same is true of $\ket{\beta_1,\ldots,\beta_r}_{\delta}$, and after the two resulting tensors are intertwined the final step is compose all the operators, yielding \eqref{eq:formula_derivative_bint}.
\end{proof}

\begin{example} For $S = 001$ we have
\begin{align*}
\den{\underline{001}}\big( \ket{\alpha}_\gamma \otimes \ket{\emptyset}_\delta \big) &= \delta \circ \alpha \circ \gamma + \delta \circ \gamma \circ \alpha\,,\\
\den{\underline{001}}\big( \ket{\alpha_1,\alpha_2}_\gamma \otimes \ket{\emptyset}_\delta \big) &= \delta \circ \alpha_1 \circ \alpha_2 + \delta \circ \alpha_2 \circ \alpha_1\,,\\
\den{\underline{001}}\big( \ket{\emptyset}_\gamma \otimes \ket{\beta}_\delta \big) &= \beta \circ \gamma \circ \gamma\,,\\
\den{\underline{001}}\big( \ket{\alpha}_\gamma \otimes \ket{\beta}_\delta \big) &= \beta \circ \alpha \circ \gamma + \beta \circ \gamma \circ \alpha\,,\\
\den{\underline{001}}\big( \ket{\alpha_1,\alpha_2}_\gamma \otimes \ket{\beta}_\delta \big) &= \beta \circ \alpha_1 \circ \alpha_2 + \beta \circ \alpha_2 \circ \alpha_1\,.
\end{align*}
and zero for all other inputs.
\end{example}

More interestingly we can also compute the derivatives of proofs of ${!}\binta_A \vdash \binta_A$. In what follows $A$ is fixed and $E = A \multimap A$.

\begin{definition} The proof $\repeat$ is
\begin{center}
\AxiomC{}
\UnaryInfC{$\textcolor{red}{{!}E} \vdash {!}E$}
\AxiomC{}
\UnaryInfC{$\textcolor{blue}{{!}E} \vdash {!}E$}
\AxiomC{}
\UnaryInfC{$\textcolor{red}{{!}E} \vdash {!}E$}
\AxiomC{}
\UnaryInfC{$\textcolor{blue}{{!}E} \vdash {!}E$}
\AxiomC{$\comp^2_A$}
\noLine\UnaryInfC{$\vdots$}
\noLine\UnaryInfC{$E, E \vdash E$}
\RightLabel{\scriptsize$\multimap L$}
\BinaryInfC{$\textcolor{blue}{{!} E}, {!}E \multimap E, E \vdash E$}
\RightLabel{\scriptsize$\multimap L$}
\BinaryInfC{$\textcolor{red}{{!} E}, \textcolor{blue}{{!} E}, \binta_A, E \vdash E$}
\RightLabel{\scriptsize$\multimap L$}
\BinaryInfC{$\textcolor{blue}{{!} E}, \textcolor{red}{{!} E}, \textcolor{blue}{{!} E}, \binta_A, {!}E \multimap E \vdash E$}
\RightLabel{\scriptsize$\multimap L$}
\BinaryInfC{$\textcolor{red}{{!} E}, \textcolor{blue}{{!} E}, \textcolor{red}{{!} E}, \textcolor{blue}{{!} E}, \binta_A, \binta_A \vdash E$}
\RightLabel{\scriptsize ctr}
\UnaryInfC{$\textcolor{red}{{!} E},\textcolor{blue}{{!} E},\textcolor{blue}{{!} E}, \binta_A, \binta_A \vdash E$}
\RightLabel{\scriptsize ctr}
\UnaryInfC{$\textcolor{red}{{!} E},\textcolor{blue}{{!} E}, \binta_A, \binta_A \vdash E$}
\doubleLine\RightLabel{\scriptsize$2\times \multimap R$}
\UnaryInfC{$\binta_A, \binta_A \vdash \binta_A$}
\doubleLine\RightLabel{\scriptsize$2\times$ der}
\UnaryInfC{${!}\binta_A, {!}\binta_A \vdash \binta_A$}
\RightLabel{\scriptsize ctr}
\UnaryInfC{${!}\binta_A \vdash \binta_A$}
\DisplayProof
\end{center}
which repeats a binary sequence in the sense that the cutting it against the promotion of $\underline{S}$ is equivalent under cut-elimination to $\underline{SS}$. In particular, $\den{\repeat}\ket{\emptyset}_{\den{\underline{S}}} = \den{\underline{SS}}$.
\end{definition}

Given $S, T \in \{0,1\}^*$ the derivative of $\repeat$ at $S$ in the direction of $T$ is
\be
\den{\repeat}\ket{\den{\underline{T}}}_{\den{\underline{S}}} \in \den{\binta_A} = \Hom_k( {!} \End_k(V) \otimes {!}\End_k(V), \End_k(V))\,,
\ee
and as promised in the Introduction:

\begin{lemma} $\den{\repeat}\ket{\den{\underline{T}}}_{\den{\underline{S}}} = \den{\underline{ST}} + \den{\underline{TS}}$.
\end{lemma}
\begin{proof}
The value of the left-hand side on a tensor $\ket{\alpha_1,\ldots,\alpha_s}_{\gamma} \otimes \ket{\beta_1,\ldots,\beta_r}_{\delta}$ is computed by reading the proof-tree for $\repeat$ from bottom to top:
\begin{align*}
\ket{\den{\underline{T}}}_{\den{\underline{S}}} 
&\xmapsto{\makebox[1cm]{\scriptsize\text{ctr}}} \ket{\den{\underline{T}}}_{\den{\underline{S}}} \otimes \ket{\emptyset}_{\den{\underline{S}}} + \ket{\emptyset}_{\den{\underline{S}}} \otimes \ket{\den{\underline{T}}}_{\den{\underline{S}}}
\\&\xmapsto{\makebox[1cm]{\scriptsize$2\times$ der}} \den{\underline{T}} \otimes \den{\underline{S}} + \den{\underline{S}} \otimes \den{\underline{T}}
\\&\xmapsto{\makebox[1cm]{\scriptsize${2\times} {R \multimap}$}} \ket{\alpha_1,\ldots,\alpha_s}_{\gamma} \otimes \ket{\beta_1,\ldots,\beta_r}_{\delta} \otimes \big(\den{\underline{T}} \otimes \den{\underline{S}} + \den{\underline{S}} \otimes \den{\underline{T}}\big)
\\&\xmapsto{\makebox[1cm]{\scriptsize$2\times$ ctr}} \sum_{I,J} \ket{\alpha_I}_{\gamma} \otimes \ket{\beta_J}_{\delta} \otimes \ket{\alpha_{I^c}}_{\gamma} \otimes \ket{\beta_{J^c}}_{\delta} \otimes \big( \den{\underline{T}} \otimes \den{\underline{S}} + \den{\underline{S}} \otimes \den{\underline{T}} \big)
\\&\xmapsto{\makebox[1cm]{}} \sum_{I,J} \den{\underline{S}}\big( \ket{\alpha_I}_{\gamma} \otimes \ket{\beta_J}_{\delta} \big) \circ \den{\underline{T}}\big( \ket{\alpha_{I^c}}_{\gamma} \otimes \ket{\beta_{J^c}}_{\delta} \big)
\\& \qquad\qquad + \sum_{I,J} \den{\underline{T}}\big( \ket{\alpha_I}_{\gamma} \otimes \ket{\beta_J}_{\delta} \big) \circ \den{\underline{S}}\big( \ket{\alpha_{I^c}}_{\gamma} \otimes \ket{\beta_{J^c}}_{\delta} \big)
\end{align*}
which agrees with $\den{\underline{ST}} + \den{\underline{TS}}$ on $\ket{\alpha_1,\ldots,\alpha_s}_{\gamma} \otimes \ket{\beta_1,\ldots,\beta_r}_{\delta}$ by Lemma \ref{lemma:derivative_bint}.
\end{proof}

\subsection{Multiplication}\label{section:mult}

The multiplication of Church numerals is encoded by a proof $\mult_A$ of ${!} \inta_A, \inta_A \vdash \inta_A$, see for example \cite[\S 2.5.2]{girard_complexity}. To construct the proof tree it will be convenient to introduce the following intermediate proof $\gamma$, writing $E = A \multimap A$ as above:
\begin{center}
\AxiomC{}
\UnaryInfC{${!}E \vdash {!}E$}
\AxiomC{}
\UnaryInfC{$E \vdash E$}
\RightLabel{\scriptsize$\multimap L$}
\BinaryInfC{${!}E, \inta_A \vdash E$}
\RightLabel{\scriptsize der}
\UnaryInfC{${!}E, {!}\inta_A \vdash E$}
\RightLabel{\scriptsize prom}
\UnaryInfC{${!}E, {!}\inta_A \vdash {!}E$}
\DisplayProof
\end{center}
Then $\den{\gamma}: {!}\End_k(V) \otimes {!}\den{\inta_A} \to {!}\End_k(V)$ is a morphism of coalgebras  such that 
\[\den{\gamma}(t \otimes \vacu_\alpha) = \vacu_{\alpha(t)} \qquad \text{and} \qquad \den{\gamma}(t \otimes \ket{v}_\alpha) = \ket{v(t)}_{\alpha(t)},\]
for $\alpha, v \in \den{\inta_A} = \Hom_k({!}\End_k(V), \End_k(V))$ and $t \in {!} \End_k(V)$. The proof $\mult_A$ is
\begin{center}
\AxiomC{$\gamma$}
\noLine\UnaryInfC{$\vdots$}
\noLine\UnaryInfC{$!E, {!}\inta_A \vdash {!}E$}
\AxiomC{}
\UnaryInfC{$E \vdash E$}
\RightLabel{\scriptsize$\multimap L$}
\BinaryInfC{${!}E, {!}\inta_A, \inta_A \vdash E$}
\RightLabel{\scriptsize$\multimap R$}
\UnaryInfC{${!}\inta_A, \inta_A \vdash \inta_A$}
\DisplayProof
\end{center}
Let $l,m,n \ge 0$ be integers. We write $\mult_A(-,n)$ for the proof of ${!} \inta_A \vdash \inta_A$ obtained from the above by cutting against the proof $\underline{n}$ of $\vdash \inta_A$. The derivative of this proof at $\alpha = \den{\underline{l}}$ in the direction of $v = \den{\underline{m}}$ is the element of $\den{\inta_A}$ given on $t \in {!} \End_k(V)$ by
\begin{align*} 
\den{\mult_A}\big(\ket{v}_\alpha \otimes \den{\underline{n}}\big)(t)
= \den{\underline{n}}\big(\ket{v(t)}_{\alpha(t)}\big)
= \sum_{i = 1}^{n} \alpha(t)^{i-1} v(t) \alpha(t)^{n-i}
\end{align*}
using Lemma \ref{lemma:nderiv} in the last step. When $t = \vacu_x$ for $x \in \End_k(V)$, this evaluates to
\[ 
\sum_{i = 1}^{n} \alpha(t)^{i-1} v(t) \alpha(t)^{n-i} 
= \sum_{i = 1}^{n} x^{l(i-1)} x^m x^{l(n-i)} 
= n x^{l(n-1)+m}.
\]
This result agrees with a more traditional calculus approach using limits:
\begin{align*}
&\lim_{h \to 0} \frac{\den{\mult_A}(\vacu_{\den{\underline{l}} + h \den{\underline{m}}} \otimes \den{n})\vacu_x  - \den{\mult_A}(\vacu_{\den{\underline{l}}} \otimes \den{\underline{n}})\vacu_x}{h}\\
&= \lim_{h \to 0} \frac{\den{\underline{n}}\vacu_{x^l + h x^m} - \den{\underline{n}}\vacu_{x^l}}{h} \\
\\&= \lim_{h \to 0} \frac{(x^l + h x^m)^n - x^{ln}}{h} 
\\&= nx^{l(n-1)+m}.
\end{align*}

\section{Differential lambda calculus}\label{section:difflambda}




Categorically speaking $\lambda$-calculi are modelled by Cartesian closed categories \cite{lambek} and any model of linear logic gives rise to a Cartesian closed category by taking the Kleisli category of the comonad interpreting the exponential connective; see \cite[\S 7]{mellies}. There is a parallel relationship between differential $\lambda$-calculus and differential linear logic, which has been worked out in the language of Cartesian differential categories \cite{cite5,cite3}. 

As above $k$ is an algebraically closed field of characteristic zero, and $\cat{V}$ is the category of $k$-vector spaces. In this section we explain how the category of cofree coalgebras over $k$ gives a model of differential $\lambda$-calculus, by
\begin{itemize}
    \item Checking that $\cat{V}$ is a differential storage category with biproducts, so that
    \item by \cite[Prop 3.2.1]{cite5} the Kleisli category $\sV_{!}$ of $(\sV, {!})$ is a Cartesian differential category,
    \item and we check that $\sV_{!}$ is a differential $\lambda$-category in the sense of \cite[Definition 4.4]{cite3}.
\end{itemize}
We have already checked that the comonad ${!}$ arising from the cofree coalgebra construction defines a coalgebra modality in the sense of \cite[Definition 2.1]{blutecs}, and that we have deriving transformations $D_V: {!}V \otimes V \lto {!}V$. By Corollary \ref{corollary:diffcat} the tuple $(\sV, {!}, \delta, d, \Delta, \varepsilon, D)$ is a differential category in the sense of \cite[Definition 2.4]{blutecs}. We now prove further that this is an example of a \emph{differential storage category}. 
\vspace{0.2cm}

For this we need to understand ${!}: \sV \lto \sV$ as a monoidal functor.

\begin{definition} Given vector spaces $V, W$, let $m_{V,W}$ be the unique morphism of coalgebras which makes the following diagram commute:
\[
\begin{tikzcd}[column sep=large,row sep=large]
{!}V \otimes {!}W \arrow[dr, swap,"{d_V \otimes d_W}"]\arrow[r, "{m_{V,W}}"]& {!}(V \otimes W) \arrow[d, "{d_{V \otimes W}}"]\\ 
& V \otimes W
\end{tikzcd}
\]
\end{definition}

\begin{lemma} The morphism of coalgebras $m_{V,W}$ is natural in $V,W$.
\end{lemma}
\begin{proof}
If $\varphi: V \lto V'$ and $\psi: W \lto W'$ are linear then the outside diagram in 
\[
\begin{tikzcd}
{!}V \otimes {!}W \arrow[ddr, "{d_V \otimes d_W}"]\arrow[rrrr, "{{!}\varphi \otimes {!}\psi}"]\arrow[dddd, swap,"{m_{V,W}}"]& && & {!}V' \otimes {!}W' \arrow[dddd, "{m_{V',W'}}"]\arrow[ddl, swap,"{d_{V'} \otimes d_{W'}}"]\\ \\
& V \otimes W \arrow[rr, "{\varphi \otimes \psi}"]&& V' \otimes W'& \\ \\
{!}(V \otimes W) \arrow[uur, swap,"{d_{V \otimes W}}"]\arrow[rrrr,swap, "{{!}(\varphi \otimes \psi)}"]& && & {!}(V' \otimes W') \arrow[uul, "{d_{V \otimes W}}"]
\end{tikzcd}
\]
commutes, by virtue of each of the sub-diagrams commuting.
\end{proof}

\begin{definition} Let $u: \kk \lto {!}\kk$ be the unique morphism of coalgebras with the property that $d_\kk \circ u = 1_\kk$. Explicitly, $u(1) = \vacu_1$.
\end{definition}

\begin{lemma}
The data $\{m_{V,W}\}_{V,W \in \cat{V}}$ and $u$ make ${!}: \sV \lto \sV$ a lax monoidal functor.
\end{lemma}

\begin{proof}
We have to check commutativity of the outer square in
\[
\begin{tikzcd}
({!}A \otimes {!}B) \otimes {!}C \arrow[dd,swap,"{m_{A,B}\otimes 1}"]\arrow[rrr,"{\alpha}"]\arrow[ddr,"{d^{\otimes 3}}"]& & & {!}A \otimes ({!}B \otimes {!}C) \arrow[dd,"{1 \otimes m_{B,C}}"]\arrow[ddl,swap,"{d^{\otimes 3}}"]\\ \\
{!}(A \otimes B) \otimes {!}C \arrow[dd,swap,"{m_{A\otimes B, C}}"]\arrow[r, "d^{\otimes 2}"] & (A \otimes B) \otimes C \arrow[r, "\alpha"]&
 A \otimes (B \otimes C)\arrow[r, "d^{\otimes 2}"] & {!}A \otimes {!}(B \otimes C)\arrow[dd,"{m_{A, B\otimes C}}"] \\ \\
{!}((A \otimes B) \otimes C)\arrow[rrr,swap,"{{!}\alpha}"]\arrow[uur,swap,"{d}"] & & & {!}(A \otimes (B \otimes C))\arrow[uul,"{d}"]
\end{tikzcd}
\]
and the outer squares in
\[
\begin{tikzcd}
{!}A \otimes \kk \arrow[rrr,"{1 \otimes u}"]\arrow[dd,swap,"{\cong}"]\arrow[drr,"{d \otimes 1}"]&&& {!}A \otimes {!}\kk \arrow[dd,"{m_{A,\kk}}"]\arrow[dl,swap,"{d \otimes d}"] \\ 
& A & A \otimes \kk\arrow[l,"{\cong}"] & \\ 
{!}A \arrow[ru,"{d}"] &&& {!}(A \otimes \kk)\arrow[lll,"{\cong}"]\arrow[lu,swap,"{d}"]
\end{tikzcd}
\]
\[
\begin{tikzcd}
\kk \otimes {!}A \arrow[rrr,"{u \otimes 1}"]\arrow[dd,swap,"{\cong}"]\arrow[drr,"{1 \otimes d}"]&&& {!}\kk \otimes {!}A \arrow[dd,"{m_{\kk,A}}"]\arrow[dl,swap,"{d \otimes d}"] \\ 
& A & \kk \otimes A\arrow[l,"{\cong}"] & \\
{!}A \arrow[ru,"{d}"] &&& {!}(\kk \otimes A)\arrow[lll,"{\cong}"]\arrow[lu,swap,"{d}"]
\end{tikzcd}
\]
but from the given decompositions of these squares, this is clear.
\end{proof}

\begin{lemma} The tuple $({!}, m, u)$ is a symmetric monoidal functor on $(\cat{V}, \otimes, k)$.
\end{lemma}

\begin{proof}
This means that the following diagram commutes
\[
\begin{tikzcd}
{!}V \otimes {!}W \arrow[rr,"{m}"]\arrow[ddd,swap,"{\sigma}"]\arrow[dr,swap,"{d \otimes d}"]&& {!}(V \otimes W)\arrow[ddd,"{{!}\sigma}"]\arrow[dl,"{d}"] \\
& V \otimes W\arrow[d,"{\sigma}"] \\
& W \otimes V \\
{!}W \otimes {!}V \arrow[ur,"{d \otimes d}"]\arrow[rr,swap,"{m}"]&& {!}(W \otimes V)\arrow[ul,swap,"{d}"]
\end{tikzcd}
\]
which is clear from the given decomposition.
\end{proof}

\begin{lemma}
The natural transformations $\delta: {!} \lto {!!}$ and $d: {!} \lto \operatorname{id}_{\cat{V}}$ are lax monoidal.
\end{lemma}

\begin{proof}
We need to check commutativity of 
\[
\begin{tikzcd}[column sep=large]
{!}A \otimes {!}B \arrow[r, "{\delta_A \otimes \delta_B}"]\arrow[dr,swap, "{1}"]\arrow[ddd, swap,"{m_{A,B}}"] & {!!}A \otimes {!!}B \arrow[d, "{d \otimes d}"]\arrow[r, "{m_{{!}A, {!}B}}"]& {!}({!}A \otimes {!}B) \arrow[dl, "{d}"]\arrow[ddd, "{{!}m_{A,B}}"]\\
& {!}A \otimes {!}B \arrow[d, "{m_{A,B}}"]&\\
& {!}(A \otimes B) & \\
{!}(A \otimes B)\arrow[ur, "{1}"]\arrow[rr,swap, "{\delta_{A \otimes B}}"] && {!!}(A \otimes B)\arrow[ul, swap,"{d}"]
\end{tikzcd}
\]
and also
\[
\begin{tikzcd}[column sep=large]
&\kk\arrow[d,"{1}"]\arrow[dr,"{u}"]\arrow[ddl,swap,"{u}"]& \\
&\kk\arrow[d,"{u}"]&{!}\kk\arrow[l,"{d}"]\arrow[ddl,"{{!}u}"] \\
{!}\kk\arrow[dr,swap,"{\delta_\kk}"]\arrow[r,"{1}"]&{!}\kk &\\
&{!!}\kk\arrow[u,"{d}"]&
\end{tikzcd}
\]
and
\[
\begin{tikzcd}
  {!}A \otimes {!}B  \arrow[rr,"{d \otimes d}"]\arrow[dd,swap,"{m_{A,B}}"] && A \otimes B   \arrow[dd,"{1}"] 
\\ \\ 
 {!}(A \otimes B)   \arrow[rr,swap,"{d_{A \otimes B}}"] &&    A \otimes B
\end{tikzcd} 
\qquad\qquad
\begin{tikzcd}
    \kk    \arrow[rr,"{u}"]\arrow[ddrr,swap,"{1}"] &&     {!}\kk   \arrow[dd,"{d}"] 
\\ \\ 
&& \kk
\end{tikzcd}
\]
which are all clear from the given decompositions.
\end{proof}

The above results together prove that:

\begin{lemma}
$({!}, \delta, d, m, u)$ is a symmetric monoidal comonad.
\end{lemma}

Now each ${!}A$ is by construction a cocommutative comonoid in the category of vector spaces, and the comultiplication $\Delta: {!}A \lto {!}A \otimes {!}A$ and counit $\varepsilon: {!}A \lto \kk$ are morphisms of ${!}$-coalgebras in the following sense. Recall that a coalgebra \cite[\S 4.1]{borceux2} for the comonad $({!}, \delta, d)$ (henceforth called a ${!}$-\emph{coalgebra}) is a vector space $V$ and linear map $\varphi: V \lto {!}V$ with the property that the diagrams
\begin{equation}\label{eqn: coalgebra for the monad bang}
\begin{tikzcd}
V \arrow[rr,"{\varphi}"]\arrow[ddrr,swap,"{1}"]&& {!}V\arrow[dd,"{d_V}"] \\ \\ && V
\end{tikzcd}
\qquad\qquad
\begin{tikzcd}
V    \arrow[rr,"{\varphi}"]\arrow[dd,swap,"{\varphi}"] &&  {!}V  \arrow[dd,"{\delta_V}"] 
\\ \\ 
{!}V    \arrow[rr,swap,"{{!}\varphi}"] &&    {!!}V
\end{tikzcd}
\end{equation}
commute. We make ${!}A \otimes {!}A$ into a ${!}$-coalgebra by
\[
\begin{tikzcd}
{!}A \otimes {!}A \arrow[rr, "{\delta_A \otimes \delta_A}"]
&&
{!!}A \otimes {!!}A \arrow[rr, "{m_{!A,!A}}"]
&&
{!}({!}A \otimes {!}A)
\end{tikzcd}
\]
and $\kk$ is a coalgebra via $u: \kk \lto {!}\kk$. Then (see \cite[Definition 4.1]{blutecs}):

\begin{lemma}
The tuple $({!}, \delta, d)$ is a storage modality on $(\sV, \otimes, \kk)$ when we take the canonical comonoid structures $({!}A, \Delta, \varepsilon)$.
\end{lemma}

\begin{proof}
We need only check $\Delta, \varepsilon$ are morphisms of ${!}$-coalgebras. In the first case consider the commutative diagram
\[
\begin{tikzcd}[column sep=large]
{!}A \otimes {!}A \arrow[r, "{\delta_A \otimes \delta_A}"]\arrow[dr,swap, "{1}"] & {!!}A \otimes {!!}A \arrow[d, "{d \otimes d}"]\arrow[r, "{m_{{!}A, {!}A}}"]& {!}({!}A \otimes {!}A) \arrow[dl, "{d}"]\\
& {!}A \otimes {!}A &\\
& {!}A \arrow[u, "{\Delta}"]& \\
{!}A\arrow[ur, "{1}"]\arrow[rr,swap, "{\delta_{A}}"]\arrow[uuu, "{\Delta}"] && {!!}A\arrow[ul, swap,"{d}"]\arrow[uuu, swap,"{{!}\Delta}"]
\end{tikzcd}
\]
Now since ${!}A$ is cocommutative $\Delta$ is a morphism of coalgebras, so we may use the universal property of ${!}({!}A \otimes {!}A)$ as indicated. The fact that $\varepsilon: {!}A \lto \kk$ is a morphism of ${!}$-coalgebras is exhibited by the diagram
\[
\begin{tikzcd}
{!!}A \arrow[dr,swap, "{d}"]\arrow[rrr, "{{!}\varepsilon}"]&&& {!}\kk\arrow[dl, "{d}"] \\
& {!}A \arrow[r, "{\varepsilon}"]& \kk & \\
{!}A \arrow[uu, "{\delta_A}"]\arrow[ur, "{1}"]\arrow[rrr, swap,"{\varepsilon}"] &&& \kk\arrow[ul,swap, "{1}"]\arrow[uu, swap,"{u}"]
\end{tikzcd}
\]
where we use that $\varepsilon$ is a morphism of coalgebras.
\end{proof}

We conclude that $\sV$ has the structure of a differential storage category, in the sense of \cite[Definition 4.10]{blutecs}.
 
\begin{theorem}
With the above structure $\sV$ is a differential storage category.
\end{theorem}

\begin{proof}
We have shown $\sV$ is an additive storage category \cite[Definition 4.4]{blutecs}, with a deriving transformation satisfying the $\nabla$-rule (Lemma \ref{lemma_nablarule}) so we are done.
\end{proof}

\subsection{Cartesian differential categories} \label{section: cartesian differential categories}

The Kleisli category of the comonad ${!}$ on $\cat{V}$ is equivalent to the category of cofree coalgebras; for the reader's convenience we recall the proof in Appendix \ref{section:kleisli}. Under this equivalence the tensor product of coalgebras, which gives the categorical product in $\Coalg_k$ by \cite[p.49, p.65]{sweedler}, corresponds to the product $(X, Y) \mapsto X \oplus Y$ on the Kleisli category. 

By \cite[Prop 3.2.1]{cite5}, the Kleisli category $\sV_{!}$ is a Cartesian differential category. In order to describe the Cartesian differential structure, recall that there is a bijection
\begin{gather*}
\xymatrix@C+2pc{ \Coalg_k( {!} X, {!} Y ) \ar[r]^-{\cong} & \Hom_k( {!} X, Y ) = \cat{V}_{!}(X,Y) }\\
F \longmapsto d_Y \circ F\,.
\end{gather*}
The Cartesian differential operator on $\sV_{!}$ prescribed by \cite{cite5} is a function
\[
\mathbb{D}_{X,Y}: \sV_{!}( X, Y ) \lto \sV_{!}( X \times X, Y )
\]
which may be viewed as a function
\[
\Hom_k({!} X, Y) \lto \Hom_k({!}(X \oplus X), Y)
\]
or equivalently as a function
\[
\Coalg_k( {!} X, {!} Y ) \lto \Coalg_k( {!} X \otimes {!} X, {!} Y )\,.
\]
In the following we write $D_X$ for both the deriving transformation ${!} X \otimes X \lto {!} X$ and the map obtained from $D_X$ by precomposing with the swap $\sigma_{!X, X}$. By definition $\mathbb{D}_{X,Y}$ assigns to a linear map $f: {!} X \lto Y$ the composite
\[
\xymatrix@R+1pc@C+1pc{
{!} X \otimes {!} X \ar[d]^-{\cong}_-{\Theta} & & & & Y\\
{!}(X \oplus X) \ar[r]_-{\Delta} & {!}(X \oplus X) \otimes {!}(X \oplus X) \ar[r]_-{{!} \pi_0 \otimes {!} \pi_1} & {!} X \otimes {!} X \ar[r]_-{d_X \otimes 1} & X \otimes {!} X \ar[r]_-{D_X} & {!} X \ar[u]_-{f}
}
\]
where $D_X$ is the deriving transformation, $\Theta$ is the canonical isomorphism of Lemma \ref{lemma:bangisadditive} and $\pi_i: X \oplus X \lto X$ denote the projections. In our case we can simplify this definition. Recall that a Cartesian differential category \cite[Definition 2.1.4]{cite5} is a Cartesian left additive category equipped with a Cartesian differential operator, denoted here by $\mathbb{D}$. The left additive structure $\operatorname{plus}(-,-)$ on $\sV_{!}$ is given by \cite[Proposition 1.3.3]{cite5} from the Cartesian and additive structure on $\sV$ (the usual $\oplus$ and $+$) as follows: given $F,G \in \Coalg_k({!} X, {!} Y)$ we define $\operatorname{plus}(F,G) \in \Coalg_k({!} X, {!} Y)$ to be the unique morphism of coalgebras with
\[
d_Y \circ \operatorname{plus}p(F,G) = d_Y \circ ( F + G )\,.
\]

\begin{lemma} The composite
\[
\xymatrix@C+2pc{
{!}X \otimes {!}X \ar[r]^-{\Theta}_{\cong} & {!}(X \oplus X) \ar[r]^-{\Delta} & {!}(X \oplus X) \otimes {!}(X \oplus X) \ar[r]^-{{!}\pi_0 \otimes {!}\pi_1} & {!}X \otimes {!}X
}
\]
is the identity map.
\end{lemma}

\begin{proof}
We can see this by direct calculation (see Definition \ref{defn:ket} for notation):
\begin{align*}
&\;\;({!}\pi_0 \otimes {!}\pi_1) \circ \Delta \circ \Theta \left(\ket{v_1, ..., v_s}_P \otimes \ket{w_1, ..., w_t}_Q\right)
\\&=
({!}\pi_0 \otimes {!}\pi_1) \circ \Delta\left(\ket{(v_1,0), ..., (v_s,0), (0,w_1), ..., (0,w_t)}_{(P,Q)}\right)
\\&= 
({!}\pi_0 \otimes {!}\pi_1)\Sum_{I \subseteq [s]} \Sum_{J \subseteq [t]}\ket{v_I, w_J}_{(P,Q)} \otimes \ket{v_{I^c}, w_{J^c}}_{(P,Q)}
\\&= 
\Sum_{I \subseteq [s]} \Sum_{J \subseteq [t]} \delta_{J = \emptyset}\ket{v_I}_{P} \otimes \delta_{I^c = \emptyset}\ket{w_{J^c}}_{Q}
\\&= 
\Sum_{I \subseteq [s]} \Sum_{J \subseteq [t]} \delta_{J^c = [t]}\delta_{I = [s]}\ket{v_I}_{P} \otimes \ket{w_{J^c}}_{Q}
\\&= 
\ket{v_1, ..., v_s}_P \otimes \ket{w_1, ..., w_t}_Q
\end{align*}
which proves the claim.
\end{proof}

The upshot is that \cite[Prop 3.2.1]{cite5} implies the following Cartesian differential operator makes $\cat{V}_{!}$ into a Cartesian differential category:

\begin{definition}
The Cartesian differential operator
\[
\mathbb{D}_{X,Y}: \sV_{!}(X,Y) \lto \sV_{!}(X \times X, Y)
\]
sends a linear map $f: {!}X \lto Y$ to the linear map
\[
\xymatrix@C+2pc{
{!}(X \oplus X) \ar[r]^{\Theta^{-1}}_-{\cong} & {!}X \otimes {!}X \ar[r]^{d_X \otimes 1} & X \otimes {!}X \ar[r]^-{D_X} & {!}X \ar[r]^{f} & Y
}
\]
\end{definition}

Direct computation shows that
\begin{align*}
\mathbb{D}_{X,Y}(f)&\left(\ket{v_1, ..., v_s}_P \otimes \ket{w_1, ..., w_t}_Q\right)\\
\qquad &= fD(d \otimes 1)\left(\ket{v_1, ..., v_s}_P \otimes \ket{w_1, ..., w_t}_Q\right)
\\&= fD(\delta_{s = 0} P \otimes \ket{w_1, ..., w_t}_Q+ \delta_{s = 1} v_1 \otimes \ket{w_1, ..., w_t}_Q)
\\&= \delta_{s = 0} f\ket{P,w_1, ..., w_t}_Q+ \delta_{s = 1} f\ket{v_1, w_1, ..., w_t}_Q.
\end{align*}
and the lifting of $\mathbb{D}_{X,Y}(f)$ to a morphism of coalgebras ${!}(X \oplus X) \lto {!} Y$ may be described explicitly using \cite[Theorem 2.22]{murfet_coalg}.




Next we show that $\sV_{!}$ together with the maps $\mathbb{D}_{X,Y}$ is a Cartesian \emph{closed} differential category in the sense of \cite[\S1.4]{cite5}. First we recall the closed structure on $\sV_{!}$. Throughout $A,B$ denote arbitrary vector spaces.

\begin{definition} Let $\Gamma_{A,B}$ denote the unique morphism of coalgebras making
\[
\begin{tikzcd}[row sep=large,column sep=large]
{!}A \otimes {!}\Hom_\kk({!}A, B) \arrow[r, "1 \otimes d"]\arrow[d,swap,dashed,"\Gamma_{A,B}"] & {!}A \otimes \Hom_\kk({!}A, B) \arrow[d, "\eval"] \\{!}B \arrow[r,swap,"d"] & B
\end{tikzcd}
\]
commute, where we take the usual coalgebra structure on the tensor product.
\end{definition}

\begin{definition} We define
\[
\HOM(A,B) = \Hom_\kk({!}A, B)\,.
\]
Given a morphism $f \in \sV_{!}(B, B')$, that is a linear map $f: {!}B \lto B'$, we define
\[
\HOM(A,f) \in \sV_{!}( \HOM(A,B), \HOM(A,B') ) = \Hom_k\big( {!} \Hom_k({!} A, B), \Hom_k({!} A, B') \big)
\]
to be the linear map corresponding under the Hom-tensor adjunction to $f \circ \Gamma_{A,B}$.
\end{definition}

Again, this may be computed explicitly using \cite[Theorem 2.22]{murfet_coalg}.


\begin{lemma}
$\HOM(A,-)$ is a functor $\sV_{!} \lto \sV_{!}$.
\end{lemma}

\begin{proof}
To show $\HOM(A, 1_B) = 1_{\HOM(A,B)}$ we have to show that $\eval \circ (1 \otimes d)$ corresponds under the Hom-tensor adjunction to the dereliction $d: {!}\Hom_\kk({!}A, B) \lto \Hom_\kk({!}A, B)$. This we can do by calculation; $\eval \circ (1 \otimes d)$ is:
\begin{equation}
\ket{v_1, ..., v_s}_P \otimes \ket{\zeta_1, ..., \zeta_t}_\alpha \mapsto \delta_{t=0} \alpha\ket{v_1, ..., v_s}_P + \delta_{t = 1} \zeta_1\ket{v_1, ..., v_s}_P,
\end{equation}
which agrees with the dereliction on ${!}\Hom_\kk({!}A, B)$.

Now suppose given linear maps $g: {!}B' \lto B''$ and $f: {!}B \lto B'$, and let $\bullet$ denote the Kleisli composition. To show that
\begin{equation} \label{eqn: HOM is functorial} 
\HOM(A,g) \bullet \HOM(A,f) = \HOM(A, g \bullet f),
\end{equation}
we first observe $g \bullet f$ is the linear map
\[
{!}B \lxto{\delta} {!!}B \lxto{{!}f} {!}B' \lxto{g} B''
\]
and the left hand side of \eqref{eqn: HOM is functorial} is the composite
\[
{!}\Hom_\kk({!}A, B) \lxto{\delta} {!!}\Hom_\kk({!}A, B) \lxto{{!}\HOM(A,f)} {!}\Hom_\kk({!}A, B')\lxto{\HOM(A,g)} \Hom_\kk({!}A, B'')
\]
which corresponds under the Hom-tensor adjunction to the left hand vertical composite in the following commutative diagram
\[
\begin{tikzcd}[row sep=large, column sep=large]
{!}A \otimes {!}\Hom_\kk({!}A, B)\arrow[d,swap, "{1 \otimes \delta}"]\arrow[dr, "1"]  \\
{!}A \otimes {!!}\Hom_\kk({!}A, B)\arrow[r, "{1 \otimes d}"]\arrow[d,swap, "{1 \otimes {!}\HOM(A,f)}"] & {!}A \otimes {!}\Hom_\kk({!}A, B)\arrow[d, "{1 \otimes \HOM(A,f)}"]\\
{!}A \otimes {!}\Hom_\kk({!}A, B')\arrow[r,"{1 \otimes d}"]\arrow[ddr,"{\Gamma_{A,B'}}"]\arrow[d,swap, "{1 \otimes \HOM(A,g)}"] & {!}A \otimes \Hom_\kk({!}A, B') \arrow[d, "{\eval}"] \\
{!}A \otimes \Hom_\kk({!}A, B'')\arrow[d,swap, "{\eval}"] & B' \\
B''& {!}B' \arrow[u, swap, "{d}"]\arrow[l, "{g}"]
\end{tikzcd}
\]
So to prove \eqref{eqn: HOM is functorial} it suffices to show
\[
\Gamma_{A,B'} \circ (1 \otimes {!}\HOM(A,f)) \circ (1 \otimes \delta) = {!}f \circ \delta \circ \Gamma_{A,B}.
\]
But both sides are morphisms of coalgebras, so we may compare them after postcomposition with $d$, and this reduces to
\[
\eval \circ (1 \otimes \HOM(A,f)) = f \circ \Gamma_{A,B},
\]
which is true by definition.
\end{proof}

\begin{lemma}
The functor $\HOM(A,-): \sV_{!} \lto \sV_{!}$ is right adjoint to $- \times A$.
\end{lemma}

\begin{proof}
Recall that $- \times A$ as a functor on $\sV_{!}$ sends $f: {!}B \lto B'$ to the linear map
\[
f \times A: {!}(B \oplus A) \cong {!}B \otimes {!}A \lxto{F \otimes 1} {!}B' \otimes {!}A \cong {!}(B' \oplus A) \lxto{d} B' \oplus A
\] 
where $F$ is the morphism of coalgebras lifting $f$, given explicitly by \cite[Theorem 2.22]{murfet_coalg}. Given $v_1, ..., v_s, P \in A$ and $\omega_1, ..., \omega_t, Q \in B$ and writing $\cat{P}_{[t]}$ for the set of partitions of $[t] = \{1,\ldots,t\}$ we may calculate $f \times A$ as the map
\begin{align*}
\ket{w_1, ..., w_t}_Q \otimes \ket{v_1, ..., v_s}_P
&\lxmapsto[2.5em]{F \otimes 1}
\sum_{C \in \cat{P}_{[t]}} \Ket{\prod_{c \in C} f\ket{w_c}_Q}_{f \vacu_Q} \otimes \ket{v_1, ..., v_s}_P
\\&\lxmapsto[2.5em]{\Theta} 
\sum_{C \in \cat{P}_{[t]}} \Ket{\prod_{c \in C} f\ket{w_c}_Q, v_1, ..., v_s}_{(f \vacu_Q, P)}
\\&\lxmapsto[2.5em]{d} 
\delta_{s = t = 0}(f \vacu_Q, P) + \delta_{s = 0, t > 0} (f\ket{w_1, ..., w_t}_Q, 0)\\
&\qquad\qquad + \delta_{s = 1, t = 0} (0, v_1).
\end{align*}
So in summary:
\begin{align*}
&(f \times A)(\ket{w_1, ..., w_t}_Q \otimes \ket{v_1, ..., v_s}_P) 
\\&\qquad\qquad=\big(\delta_{s = t = 0}f \vacu_Q + \delta_{s = 0, t > 0}f\ket{w_1, ..., w_t}_Q, \delta_{s = t = 0} P + \delta_{s = 1, t = 0} v_1\big).
\end{align*}
We have bijections for vector spaces $A,B,C$
\begin{align*}
\cCoalg_\kk({!}C \otimes {!}A, {!}B) 
&\cong \Hom_\kk({!}C \otimes {!}A, B)
\\&\cong \Hom_\kk({!}C, \Hom_\kk({!}A, B))
\\&\cong \cCoalg_\kk({!}C, {!}\Hom_\kk({!}A, B))
\\&= \cCoalg_\kk({!}C, {!}\HOM(A, B)),
\end{align*}
and hence a bijection
\begin{equation}
\sV_{!}(C \times A, B) \cong \sV_{!}(C, \HOM(A,B)).
\end{equation}
The question that remains is whether these bijections are natural in $C, B$. Clearly they are natural in $C$. To prove naturality in $B$ we have to show that for a linear map $f: {!}B \lto B'$
\[
\begin{tikzcd}[row sep=large, column sep=large]
  \sV_{!}(C \times A, B)  \arrow[r,"{\cong}"]\arrow[d,swap,"{\sV_{!}(1,f)}"] &  \sV_{!}(C, \HOM(A,B))  \arrow[d,"{\sV_{!}(1,\HOM(A,f))}"] 
\\ 
 \sV_{!}(C \times A, B')   \arrow[r,swap,"{\cong}"] &  \sV_{!}(C, \HOM(A,B')) 
\end{tikzcd}
\]
commutes. That is, given a morphism of coalgebras $\gamma: {!}C \otimes {!}A \lto {!}B$ we have to show
\be\label{eq:some_diagram}
\HOM(A,f) \circ \prom(\widetilde{d \circ \gamma}) = \widetilde{f \circ \gamma}: {!}C \lto \Hom_\kk({!}A, B')
\ee
where $\widetilde{z}$ denotes the morphism corresponding to $z$ under the Hom-tensor adjunction. For this consider the diagram

\[
\begin{tikzcd}
{!}A \otimes {!}C \arrow[dddddd, bend right=60, swap,"{\gamma}"] \arrow[rdd, bend left=18, "{1 \otimes \prom(\widetilde{d \circ \gamma})}"]\arrow[dd, "{1 \otimes \widetilde{d \circ \gamma}}"]    \\ \\
{!}A \otimes \Hom_\kk({!}A, B) \arrow[dd,swap, "{\eval}"] & {!}A \otimes {!}\Hom_\kk({!}A, B) \arrow[dd, "{1 \otimes \HOM(A,f)}"]\arrow[l,swap, "{1 \otimes d}"]\arrow[ldddd, bend left=0,swap, "{\Gamma_{A,B}}"]\\ \\
B  & {!}A \otimes \Hom_\kk({!}A, B')\arrow[dd, "{\eval}"] \\ \\
{!}B \arrow[uu, "{d}"]\arrow[r, bend right=0, "{f}"]& B'
\end{tikzcd}
\]
From the calculation
\begin{align*}
d \circ \Gamma_{A,B} \circ (1 \otimes \prom(\widetilde{d \circ \gamma})) &= \eval \circ (1 \otimes d) \circ \prom(\widetilde{d \circ \gamma})\\
&= \eval \circ (1 \otimes \widetilde{d \circ \gamma})\\
&= d \circ \gamma
\end{align*}
we deduce that $\Gamma_{A,B} \circ (1 \otimes \prom(\widetilde{d \circ \gamma})) = \gamma$ since both sides are morphisms of coalgebras. From this and the above diagram we easily deduce \eqref{eq:some_diagram}.
\end{proof}


\begin{lemma}
With the above structure $\sV_{!}$ is a Cartesian closed left additive category.
\end{lemma}
\begin{proof}
We need to show that 
\[
\sV_{!}(A \times B, C) \lxto{\cong} \sV_{!}(A, \HOM(B,C))
\]
is an isomorphism of monoids. But this map is the Hom-tensor adjunction
\[
\Hom_\kk({!}A \otimes {!}B, C) \lxto{\cong} \Hom_\kk({!}A, \Hom_\kk({!}B, C))
\]
which is linear, so this is clear.
\end{proof}

\begin{theorem}\label{theorem:cofree_difflambda}
$\sV_{!}$ is a differential $\lambda$-category \cite[Definition 4.4]{cite3} and thus a model of the simply-typed differential $\lambda$-calculus \cite[\S4.3]{cite3}.
\end{theorem}

\begin{proof}
First we observe that $\sV_{!}$ is a Cartesian closed differential category in the sense of \cite[Definition 4.2]{cite3}. It is a Cartesian closed left additive category, and we have already observed in Section \ref{section: cartesian differential categories} it has an operator $\mathbb{D}_{X,Y}(-)$ satisfying the axioms of a Cartesian differential category. It remains to check the axiom (D-Curry) which says given $f: C \times A \lto B$ in $\sV_{!}$ and denoting currying by $\Lambda$, that
\begin{equation}
    \mathbb{D}(\Lambda f) = \Lambda(\mathbb{D}(f) \circ \<\pi_1 \times 0_A, \pi_2 \times 1_A\>).
\end{equation}
Here $\Lambda f: C \lto \HOM(A,B)$ and so $\mathbb{D}(\Lambda f): C \times C \lto \HOM(A,B)$, whereas the right hand side corresponds under adjunction to
\begin{equation}
    (C \times C) \times A \lxto[5em]{\<\pi_1 \times 0_A, \pi_2 \times 1_A\>} (C \times A) \times (C \times A) \lxto{\mathbb{D}(f)} B.
\end{equation}
In $\cat{V}$ this map is the composition of (where $0_A$ denotes the lift of $0: {!} A \lto A$)
\[
\xymatrix@C+2pc{
{!} C \otimes {!} C \otimes {!} A \ar[d]_-{\cong} & {!} C \otimes {!} A \otimes {!} C \otimes {!}A\\
{!}(C \oplus C) \otimes {!} A \ar[d]_-{\Delta \otimes \Delta}\\
{!}(C \oplus C) \otimes {!}(C \oplus C) \otimes {!} A \otimes {!} A \ar[r]_-{\cong} & {!}(C \oplus C) \otimes {!} A \otimes {!}(C \oplus C) \otimes {!} A \ar[uu]_-{{!} \pi_1 \otimes {!} 0_A \otimes {!} \pi_2 \otimes 1_{{!}A}}
}
\]
with $\mathbb{D}(f)$ which is
\[
\xymatrix@C+2pc{
{!} C \otimes {!} A \otimes {!} C \otimes {!}A \ar[d]^-{\cong} & B\\
{!}(C \oplus A) \otimes {!}(C \oplus A) \ar[d]^-{d_{C \oplus A} \otimes 1}\\
(C \oplus A) \otimes {!}(C \oplus A) \ar[r]^-{D_{C \oplus A}} & {!}(C \oplus A) \ar[uu]_-{f}
}
\]
As above we write $[n] = \{1,\ldots,n\}$. This composite is the linear map ${!} C \otimes {!} C \otimes {!} A \lto B$ given by the formula
\begin{align*}
\ket{\alpha_1, ..., \alpha_r}_P &\otimes \ket{\beta_1, ..., \beta_s}_Q \otimes \ket{\gamma_1,\ldots,\gamma_t}_R\\
&\longmapsto \ket{\alpha_1, ..., \alpha_r,\beta_1, ..., \beta_s}_{(P,Q)} \otimes \ket{\gamma_1,\ldots,\gamma_t}_R\\
&\lxmapsto[2.5em]{\Delta \otimes \Delta} \sum_{A \subseteq [r]}\sum_{B \subseteq [s]} \sum_{C \subseteq [t]} \ket{\alpha_A, \beta_B}_{(P,Q)} \otimes \ket{\alpha_{A^c}, \beta_{B^c} }_{(P,Q)} \otimes \ket{\gamma_C}_R \otimes \ket{\gamma_{C^c}}_R\\
&\longmapsto \Sum_{A,B,C} \ket{\alpha_A, \beta_B}_{(P,Q)} \otimes \ket{\gamma_{C}}_R \otimes \ket{\alpha_{A^c}, \beta_{B^c} }_{(P,Q)} \otimes \ket{\gamma_C^c}_R
\end{align*}
\begin{align*}
&\longmapsto \Sum_{A,B,C} \delta_{B = \emptyset}\ket{\alpha_A}_{P} \otimes \delta_{C = \emptyset}\ket{\emptyset}_0 \otimes \delta_{A^c = \emptyset} \ket{\beta_{B^c} }_{Q} \otimes \ket{\gamma_C^c}_R\\
&= \ket{\alpha_1,\ldots,\alpha_r}_P \otimes \ket{\emptyset}_0 \otimes \ket{\beta_1,\ldots,\beta_s}_Q \otimes \ket{\gamma_1,\ldots,\gamma_t}_R\\
&\longmapsto \ket{\alpha_1,\ldots,\alpha_r}_{(P,0)} \otimes \ket{\beta_1,\ldots,\beta_s,\gamma_1,\ldots,\gamma_t}_{(Q,R)}\\
&\lxmapsto[2.5em]{d \otimes 1} \delta_{r=0} (P,0) \otimes \ket{\beta_1,\ldots,\beta_s,\gamma_1,\ldots,\gamma_t}_{(Q,R)}\\
&\qquad + \delta_{r=1} \alpha \otimes \ket{\beta_1,\ldots,\beta_s,\gamma_1,\ldots,\gamma_t}_{(Q,R)}\\
&\lxmapsto[2.5em]{D} \delta_{r=0} \ket{P, \beta_1,\ldots,\beta_s,\gamma_1,\ldots,\gamma_t}_{(Q,R)}\\
& \qquad + \delta_{r=1} \ket{\alpha, \beta_1,\ldots,\beta_s,\gamma_1,\ldots,\gamma_t}_{(Q,R)}\\
&\lxmapsto[2.5em]{f} \delta_{r=0} f\ket{P, \beta_1,\ldots,\beta_s,\gamma_1,\ldots,\gamma_t}_{(Q,R)}\\
&\qquad + \delta_{r=1} f\ket{\alpha, \beta_1,\ldots,\beta_s,\gamma_1,\ldots,\gamma_t}_{(Q,R)}\,.
\end{align*}
On the other hand $\mathbb{D}(\Lambda f)$ is the linear map
\[
\xymatrix@C+2pc{
{!} C \otimes {!} C \ar[r]^-{d \otimes 1} & C \otimes {!} C \ar[r]^-{D} & {!} C \ar[r]^-{\Lambda f} & \Hom_k({!} A, B)
}
\]
which sends $\ket{\alpha_1, ..., \alpha_r}_P \otimes \ket{\beta_1, ..., \beta_s}_Q$ to the linear map
\begin{align*}
\ket{\gamma_1,\ldots,\gamma_t}_R \longmapsto f\Big( \big( \delta_{r=0} \ket{P, \beta_1,\ldots,\beta_s}_Q + \delta_{r=1} \ket{\alpha, \beta_1,\ldots,\beta_s}_Q \big) \otimes \ket{\gamma_1,\ldots,\gamma_t}_R \Big)
\end{align*}
This clearly corresponds under adjunction to the linear map ${!} C \otimes {!} C \otimes {!} A \lto B$ computed above, completing the proof.
\end{proof}

\appendix

\section{Kleisli categories and coalgebras}\label{section:kleisli}

For background on (co)monads and their (co)algebras see \cite[\S 4.1]{borceux2}. 
Let ${!}: \cat{V} \lto \cat{V}$ be the comonad discussed above, and $\delta: {!} \lto {!}{!}$ the natural transformation of Definition \ref{defn:delta_map}.

\begin{lemma}
Given a vector space $V$ there is a bijection between cocommutative counital coalgebra structures on $V$ and ${!}$-coalgebra structures on $V$.
\end{lemma}

\begin{proof}
Suppose $(V, \Delta, \varepsilon)$ is such a coalgebra. There is a unique morphism of coalgebras $\varphi_{\Delta, \varepsilon}: V \lto {!}V$ such that $d_V \circ \varphi_{\Delta, \varepsilon} = 1_V$, and moreover
\[
\begin{tikzcd}[row sep=large, column sep=large]
  V  \arrow[r,"{\varphi_{\Delta, \varepsilon}}"]\arrow[d,swap,"{\varphi_{\Delta, \varepsilon}}"] & {!}V   \arrow[d,"{\delta_V}"] 
\\
  {!}V  \arrow[r,swap,"{{!}\varphi_{\Delta, \varepsilon}}"] &    {!!}V
\end{tikzcd}
\]
is easily seen to commute, so $(V, \varphi_{\Delta, \varepsilon})$ is a ${!}$-coalgebra. Then $(\Delta, \varepsilon) \mapsto \varphi_{\Delta, \varepsilon}$ is our claimed bijection. Given a ${!}$-coalgebra $(V, \varphi)$ define $\Delta$ as the composite
\[
\begin{tikzcd}
V \arrow[rr, "{\varphi}"]&& {!}V \arrow[rr, "{\Delta}"]&& {!}V \otimes {!}V \arrow[rr, "{d \otimes d}"]&& V \otimes V,
\end{tikzcd}
\]
and $\varepsilon$ as
\[
\begin{tikzcd}
V \arrow[rr, "{\varphi}"]&& {!}V \arrow[rr, "{\varepsilon}"]&& \kk.
\end{tikzcd}
\]
From commutativity of
\[
\begin{tikzcd}
V \arrow[dr, swap, "{1}"]\arrow[r,"{\varphi}"]& {!}V\arrow[d,"{d}"] \\
& V
\end{tikzcd}
\qquad\qquad
\begin{tikzcd}
V \arrow[r,"{\varphi}"]\arrow[d,"{\varphi}"]& {!}V\arrow[d,"{\delta_V}"] \\
{!}V \arrow[r,"{{!}\varphi}"]& {!!}V
\end{tikzcd}
\]
we deduce commutativity of
\[
\begin{tikzcd}
V \arrow[r,"{\varphi}"]\arrow[d,swap,"{\varphi}"]& {!}V \arrow[r,"{\Delta}"]\arrow[d,swap,"{\delta_V}"]& {!}V \otimes {!}V\arrow[d,swap,"{\delta_V \otimes \delta_V}"]\arrow[ddr,bend left,"{1}"] \\
{!}V \arrow[r,"{{!}\varphi}"]\arrow[d,swap,"{\Delta}"]& {!!}V \arrow[r,"{\Delta}"]& {!!}V \otimes {!!}V \arrow[dr,"{d \otimes d}"]\\
{!}V \otimes {!}V\arrow[urr,"{{!}\varphi \otimes {!}\varphi}"]\arrow[rr,swap,"{d \otimes d}"] && V \otimes V\arrow[r,swap,"{\varphi \otimes \varphi}"]  & {!}V \otimes {!}V
\end{tikzcd}
\]
That is,
\be\label{eq:useful_eq}
\Delta \circ \varphi = (\varphi \otimes \varphi) \circ (d \otimes d) \circ \Delta \circ \varphi
\ee
We use this to prove coassociativity of $(V, \Delta, \varepsilon)$ as follows. We must show that
\[
\begin{tikzcd}[column sep=large]
V \arrow[r, "{\varphi}"]\arrow[d,swap, "{\varphi}"]& {!}V\arrow[r, "{\Delta}"] & {!}V \otimes {!}V\arrow[r, "{d \otimes d}"] & V \otimes V \arrow[d, "{\varphi \otimes 1}"]\\
{!}V \arrow[d,swap, "{\Delta}"]&&& {!}V \otimes V \arrow[d, "{\Delta \otimes 1}"]\\
{!}V \otimes {!}V \arrow[d,swap, "{d \otimes d}"]&&& {!}V \otimes {!}V \otimes V\arrow[d, "{d \otimes d \otimes 1}"] \\
V \otimes V \arrow[r,swap, "{1 \otimes \varphi}"]& V \otimes {!}V \arrow[r,swap, "{1 \otimes \Delta}"]& V \otimes {!}V \otimes {!}V \arrow[r,swap, "{1 \otimes d \otimes d}"]& V \otimes V \otimes V
\end{tikzcd}
\]
commutes, or stated differently, that the two ways around the outside of the following diagram agree when precomposed with $(d \otimes d) \circ \Delta \circ \varphi$:
\[
\begin{tikzcd}
V \otimes V \arrow[r,"{1 \otimes \varphi}"]\arrow[ddd,swap,"{\varphi \otimes 1}"]& V \otimes {!}V \arrow[rr,"{1 \otimes \Delta}"]\arrow[d,swap,"{\varphi \otimes 1}"]&& V \otimes {!}V \otimes {!}V\arrow[ddd,"{1 \otimes d \otimes d}"] \\
& {!}V \otimes {!}V\arrow[r,"{1 \otimes \Delta}"]\arrow[dr,bend right,swap,"{\Delta \otimes 1}"]\arrow[dr,phantom,"(\dagger)"] & {!}V \otimes {!}V \otimes {!}V\arrow[ur,swap,"{d \otimes 1 \otimes 1}"] \\
&& {!}V \otimes {!}V \otimes {!}V\arrow[d,"{1 \otimes 1 \otimes d}"]\arrow[u,swap,"{1}"] \\
{!}V \otimes V\arrow[uur,"{1 \otimes \varphi}"]\arrow[rr,swap,"{\Delta \otimes 1}"] && {!}V \otimes {!}V \otimes V\arrow[r,swap,"{d \otimes d \otimes 1}"] & V \otimes V \otimes V
\end{tikzcd}
\]
In this diagram every square but the one marked $(\dagger)$ commutes. Now precomposing both ways around this diagram with $(d \otimes d) \circ \Delta \circ \varphi$ amounts to precomposing the two ways around the triangle $(\dagger)$ with the right hand side of \eqref{eq:useful_eq} and therefore to precomposing with $\Delta \circ \varphi$. But by coassociativity of $\Delta$, the two ways around $(\dagger)$ agree when precomposed with $\Delta$. The usphot is that the above diagram commutes, when precomposed with $(d \otimes d) \circ \Delta \circ \varphi$.

For counitality we need
\[
\begin{tikzcd}
V \arrow[r, "{\varphi}"]\arrow[dd,swap, "{1}"]& {!}V \arrow[r, "{\Delta}"]\arrow[d, swap,"{1}"]& {!}V \otimes {!}V \arrow[r, "{d \otimes d}"]\arrow[d,swap, "{\varepsilon \otimes 1}"]\arrow[dr, swap,"{1 \otimes d}"]& V \otimes V \arrow[d, "{\varphi \otimes 1}"]\\
& {!}V\arrow[r, "{\cong}"]\arrow[dl, "{d}"] & \kk \otimes {!}V \arrow[dr,swap, "{1 \otimes d}"]& {!}V \otimes V\arrow[d, "{\varepsilon \otimes 1}"] \\
V \arrow[rrr, "{\cong}"]&&& \kk \otimes V
\end{tikzcd}
\]
to commute, which is clear from the given decomposition. For cocommutativity of $\Delta$ we need commutativity of
\[
\begin{tikzcd}
V \arrow[r, "{\varphi}"]\arrow[dr,swap, "{\varphi}"]& {!}V \arrow[r, "{\Delta}"]\arrow[d, "{1}"]& {!}V \otimes {!}V \arrow[r, "{d \otimes d}"]\arrow[d, "{\sigma}"]& V \otimes V\arrow[d, "{\sigma}"] \\
& {!}V \arrow[r,swap, "{\Delta}"]& {!}V \otimes {!}V \arrow[r,swap, "{d \otimes d}"]& V \otimes V
\end{tikzcd}
\]
which is again clear. We have now assigned to any ${!}$-coalgebra structure $\varphi$ on $V$ a coalgebra structure $\Delta_\varphi, \varepsilon_\varphi$. Next we observe that the assignments 
\[
(\Delta, \varepsilon) \longmapsto \varphi_{\Delta, \varepsilon}, \qquad \varphi \longmapsto (\Delta_\varphi, \varepsilon_\varphi)
\]
are mutually inverse. Clearly given $(\Delta, \varepsilon)$ the diagrams
\[
\begin{tikzcd}
V \arrow[r, "{\varphi_{\Delta, \varepsilon}}"]\arrow[d,swap, "{\Delta}"] & {!}V \arrow[r, "{\Delta}"] & {!}V \otimes {!}V \arrow[r, "{d \otimes d}"]& V \otimes V \\
V \otimes V \arrow[urr,swap, "{\varphi_{\Delta, \varepsilon} \otimes \varphi_{\Delta, \varepsilon}}"]\arrow[urrr, bend right=15, swap, "{1}"]
\end{tikzcd}
\qquad\qquad
\begin{tikzcd}
V \arrow[r, "{\varphi_{\Delta, \varepsilon}}"]\arrow[rr,swap,bend right, "{\varepsilon}"] & {!}V \arrow[r, "{\varepsilon}"] & \kk
\end{tikzcd}
\]
show that $\Delta_{\varphi_{\Delta, \varepsilon}} = \Delta$ and $\varepsilon_{\varphi_{\Delta, \varepsilon}} = \varepsilon$. Given $\varphi$ we observe that $\varphi: (V, \Delta_\varphi, \varepsilon_\varphi) \lto {!}V$ is a morphism of coalgebras and satisfies $d \circ \varphi = 1$, which completes the proof.
\end{proof}

\begin{lemma} \label{lemma: morphism of bang coalg iff morphism of coalg}
Let $(V, \varphi), (W, \psi)$ be ${!}$-coalgebras and $(V, \Delta_\varphi, \varepsilon_\varphi), (W, \Delta_\psi, \varepsilon_\psi)$ the associated coalgebras. A linear map $f: V \lto W$ is a morphism of ${!}$-coalgebras if and only if it is a morphism of coalgebras.
\end{lemma}

\begin{proof}
To say $f$ is a morphism of ${!}$-algebras is to say
\[
\begin{tikzcd}
  V  \arrow[rr,"{\varphi}"]\arrow[dd,swap,"{f}"] && {!}V   \arrow[dd,"{{!}f}"] 
\\ \\ 
 W   \arrow[rr,swap,"{\psi}"] &&    {!}W
\end{tikzcd} \tag{*}
\]
commutes, whereas to say $f$ is a morphism of coalgebras is to say
\[
\begin{tikzcd}
  V  \arrow[rr,"{\Delta_\varphi}"]\arrow[dd,swap,"{f}"] && V \otimes V   \arrow[dd,"{f \otimes f}"] 
\\ \\ 
 W   \arrow[rr,swap,"{\Delta_\psi}"] &&    W \otimes W
\end{tikzcd}
\qquad\qquad
\begin{tikzcd}
   V     \arrow[dd,"{f}"]\arrow[rr,"{\varepsilon_\varphi}"] &&    \kk   
\\ \\ 
W\arrow[uurr,swap,"{\varepsilon_\psi}"] 
\end{tikzcd} \tag{$\dagger$}
\]
commute. Now if $(*)$ commutes it is clear that the diagrams in $(\dagger)$ commute. Conversely suppose the diagrams in $(\dagger)$ commute. We know then that $\varphi, \psi$ are morphisms of coalgebras, so it suffices to check $(*)$ after post composition by $d$.
\end{proof}

Recall that for a comonad $T$ on a category $\sC$ we have the Kleisli category $\sC_T$ and the Eilenberg-Moore category $\sC^T$ \cite[\S 4.1]{borceux2}. Note that a comonad on $\cat{C}$ is the same thing as a monad on $\cat{C}^\op$. Explicitly, $\sC_T$ is the category with
\begin{itemize}
    \item $\ob \sC_T = \ob \sC$,
    \item $\sC_T(x,y) = \sC(Tx, y)$,
    \item $\id_x^{\sC_T} \in \sC(Tx, x)$ is the counit $\varepsilon$,
    \item $\sC_T(y,z) \times \sC_T(x,y) \lto \sC_T(x,z)$ is $(g,f) \mapsto g \circ T(f) \circ \delta_X$, where $\delta: T \lto TT$.
\end{itemize}
while $\sC^T$ is category of coalgebras for $T$ as defined in (\ref{eqn: coalgebra for the monad bang}).
There is a fully faithful functor $\sC_T \lto \sC^T$ defined by $X \mapsto (TX, \delta_X)$. 

\begin{proposition}
$\sV^{!}$ is isomorphic to $\cCoalg_\kk$.
\end{proposition}

\begin{proof}
We define $F: \sV^{!} \lto \cCoalg_\kk$ by $F(V, \varphi) = (V, \Delta_\varphi, \varepsilon_\varphi)$ and $G: \cCoalg_\kk \lto \sV^{!}$ by $G(V, \Delta, \varepsilon) = (V, \varphi_{\Delta, \varepsilon})$ in the above notation. On morphism sets both of $F$ and $G$ are the identity (Lemma \ref{lemma: morphism of bang coalg iff morphism of coalg}). Clearly $F \circ G = 1$ and $G \circ F = 1$.
\end{proof}

\begin{corollary}
$\sV_{!}$ is equivalent to the full subcategory of cofree coalgebras in $\cCoalg_\kk$.
\end{corollary}

\begin{proof}
This is immediate from the above but can also be seen directly using
\[
\cCoalg_\kk({!}A, {!}B) \cong \Hom_k({!}A, B) \cong \sV_{!}(A,B).
\]
\end{proof}

\bibliographystyle{amsalpha}
\providecommand{\bysame}{\leavevmode\hbox to3em{\hrulefill}\thinspace}
\providecommand{\href}[2]{#2}

\end{document}